\documentclass[lettersize,journal]{IEEEtran}
\usepackage{amsmath,amsfonts}
\usepackage{algorithm}
\usepackage{array, framed}
\usepackage{subcaption}
\usepackage{amsthm}

\usepackage{stfloats}
\usepackage{url}
\usepackage{verbatim}
\usepackage{graphicx}
\usepackage{cite}
\usepackage{nicefrac}
\usepackage{siunitx}
\usepackage{booktabs}
\usepackage{
  color,
  float,
  epsfig,
  wrapfig,
  graphics,
  graphicx,
  subcaption
}
\usepackage{listings}
\usepackage{xcolor}
\usepackage{amsmath}
\usepackage{array}
\usepackage{caption}

\lstdefinestyle{ieee}{
  basicstyle=\ttfamily\small,
  keywordstyle=\color{blue}\bfseries,
  commentstyle=\color{gray},
  stringstyle=\color{purple},
  numbers=left,
  numberstyle=\tiny,
  numbersep=5pt,
  showstringspaces=false,
  breaklines=true,
  frame=single,
  captionpos=b
}
\usepackage{textcomp,amssymb}
\usepackage{setspace}
\usepackage{latexsym,fancyhdr,url}
\usepackage{enumerate}
\usepackage{algpseudocode}
\usepackage{graphics}
\usepackage{xparse} % argument parsing -- \edist
\usepackage{xspace}
\usepackage{multirow}
\usepackage{csvsimple}
\usepackage{balance}
\usepackage{hyperref}
\usepackage{adjustbox}
\usepackage{booktabs}

\usepackage{pifont}
\usepackage{
  tikz,
  pgfplots,
  pgfplotstable
}
\usepackage{float}
\usetikzlibrary{circuits.logic.US,calc}
\usetikzlibrary{
  shapes.geometric,
  arrows,
  external,
  pgfplots.groupplots,
  matrix
}
\pgfplotsset{compat=1.9}
\newcommand{\cmark}{\ding{51}} % ✓
\newcommand{\xmark}{\ding{55}} % ✗

\usepackage{mathtools,}

\DeclareMathAlphabet{\mathcal}{OMS}{cmsy}{m}{n}
\begin{document}

\author{Nilotpola Sarma,  Vaishali Ghanshyam Chaudhuri, and Chandan Karfa}

% The paper headers
\markboth{Journal of \LaTeX\ Class Files,~Vol.~14, No.~8, August~2021}%
{Shell \MakeLowercase{\textit{et al.}}: A Sample Article Using IEEEtran.cls for IEEE Journals}

\DeclareGraphicsExtensions{%
    .png,.PNG,%
    .pdf,.PDF,%
    .jpg,.mps,.jpeg,.jbig2,.jb2,.JPG,.JPEG,.JBIG2,.JB2}

\title{Controller–Datapath Aware Verification of Masked Hardware Generated via High-Level Synthesis}
\maketitle

\begin{abstract}

  Masking is a countermeasure against Power Side-Channel Attacks (PSCAs) in both software and hardware implementations of cryptographic algorithms. Compared to software masking, implementing masked hardware is time-consuming and error prone. Recent approaches, therefore, rely on High-Level Synthesis (HLS) tools to automatically generate masked Register-Transfer Level (RTL) hardware from verified masked software, significantly reducing design effort. Since HLS was never developed for security, HLS optimizations may impact PSCA security of the generated RTL. As a result, verifying the PSCA-security of HLS-generated masked RTL is crucial. Existing hardware masking verification tools can verify masked hardware, but may produce false positives when applied to HLS-generated designs with controller–datapath architectures obtained due to resource-shared datapath obtained via HLS. 
    This work proposes a hardware masking verification strategy for HLS-generated masked hardware.
    Our toolflow MaskedHLSVerif, performs state-wise formal verification of controller–datapath RTL obtained via HLS, thereby avoiding false positives caused by resource-shared datapaths. Our tool flow correctly verifies standard cryptographic benchmarks, consisting of cascaded masked gadgets and the PRESENT S-box masked with gadgets, where existing tools like REBECCA \cite{rebeccabloem2018formal} reports false positives. The proposed tool-flow is able to detect masking flaws induced by HLS-optimizations as well.

\end{abstract}


% \section{Introduction}
\section{Introduction}
\label{sec:intro}

 Embedded devices implementing a cryptographic algorithm are susceptible to Power Side-Channel Attacks (PSCAs) \cite{kocherDPA}, where an attacker uses the target device's power consumption information to extract the secret values processed by the underlying cryptographic algorithm. Masking \cite{dom2016} is a common countermeasure against PSCAs. Masking splits the secret information processed in the algorithm into random shares. Thereafter, execution proceeds by processing these shares independently, taking care to \textit{re-randomize} computations that cause their recombination. Masking can be applied at the hardware \cite{dom2016} and software levels \cite{ProvablySMOAES2004}. In addition to ensuring the correct and independent sharing of intermediate computations as required by software masking, hardware masking must also be resistant to glitches and transitions \cite{glitch2005mangard}. Also, masking in general, is applied manually by careful examination of design parameters. Thus masking a cryptographic hardware requires expertise across cryptography, digital design, design optimization and verification. There are several valid masked versions of any cryptographic design for a given security order, differing in the latency and number of random bits needed for re-masking \cite{ghpc, HPC2020, smoothpassage}. Most of these optimizations are design-specific \cite{cotg} and masking-scheme specific \cite{smoothpassage}. Therefore, manually masking a design is time-consuming and  error prone. 
 %Thus synthesizing optimally masked hardware is time consuming and difficult. 

High-Level Synthesis (HLS) tools automatically convert a C specification to Register-Transfer Level (RTL) and optimize latency and/or area in significantly less time \cite{vivado}. HLS can automatically generate an RTL masked design in Verilog or VHDL from a masked software specification of a cryptographic algorithm. Thus, given the difficulty of manually masking hardware, recent trends focus on using HLS to generate masked hardware from masked software \cite{shortestpath2021, maskedhls}.  
 This results in reduced design time, enabling Design-Space Exploration (DSE). 
  
  However, HLS was never developed with security in mind \cite{pundir2022analyzing}. The objective of HLS is to generate a functionally correct and efficient implementation.
 As a result, several scenarios have been discovered where the HLS optimizations have led to PSCA-vulnerabilities in the synthesized design \cite{shortestpath2021, maskedhls} while keeping the functional correctness intact. These optimizations are driven by a target latency or resource constraint provided by the user during synthesis. They may also be automatically applied by the HLS tool without any user provided directives. A few such optimizations are expression-balancing, re-association, etc. It may be noted that there are many such optimizations and their consequence on the PSCA-security of HLS-generated masked hardware is yet to be explored. Therefore, it is necessary to verify the PSCA-security of HLS-generated masked hardware.

% Failure: reassociation //in results

 % even for manual designs fv is important
There are many existing hardware masking verification tools that can verify the PSCA security of masked designs \cite{rebeccabloem2018formal, verica, pengfei2020formal, gigerl2021coco, hadzic2021cocoalma}. {\it We tried to use them to verify the security of HLS-generated masked hardware, albeit unsuccessfully. This is due to the unique datapath generated by HLS. Given a time-division multiplexed datapath (for resource-sharing) generated by HLS under a resource constraint, many combinations of the datapath operations will likely never be executed. When an existing formal verification tool is used to verify such a datapath, it will end up checking the security for all these possible combinations, including those that are never going to get executed. Thus, there will be false positives and, hence, inaccurate security estimation of the design.}

 To the best of our knowledge, verification of HLS generated masked designs was never addressed in any previous formal verification tools for masked hardware. Moreover, the State-Of-The-Art (SOTA) tools fail to verify them and produces false-positive results. In this work, we demonstrate how we came up with a tool-flow for verifying HLS-generated masked hardware and overcame this problem. This toolflow proposed the usage of SOTA tools in a way that could verify HLS-generated masked hardware without the aforementioned false positives by abstracting out the \textit{active datapath} at each state of the design's Finite State Machine (FSM).
 %and verifying them using an the existing verification tool in the loop. 
 Specifically, the  contributions of our work are as follows: 
 \begin{itemize}
     \item We identified the limitations of the SOTA hardware masking verification tools in verifying HLS generated masked hardware.      
     \item We have discussed the impact of optimizations of various steps of HLS on the security of masked hardware.
     \item We developed a formal verification method, MaskedHLSVerif that can verify the correctness of masking security of HLS generated hardware.
     \item A thorough experiment with six benchmark including cascaded AND gadgets and PRESENT cipher's S-box masked with gadgets, demonstrate the usefulness of our method. Our verification method can correctly identify masking flaws due to HLS-induced optimizations.
 \end{itemize}
 The rest of the paper is organized as follows: The motivation of work with example is presented in Section \ref{sec:motivation}. The background and related work are discussed in Section \ref{sec:background}. The impact of HLS on masking security is given in Section \ref{sec:impactofhls}. Our verification method is presented in Section \ref{sec:methodology}. Detailed experimental results are given in Section \ref{sec:results}. The paper is concluded in Section \ref{sec:conclusion}. 
\section{Motivation}
\label{sec:motivation}

Formal verification tools for masked hardware \cite{rebeccabloem2018formal, gigerl2021coco} abstract away control flow and conservatively assume that all syntactically possible dataflow paths may be active. Thus they check all possible input combinations to every operator in the design at each clock cycle. For HLS-generated masked designs, the scheduling of operations (which operation is executed in each clock cycle) are decided by resource constraints, which enable HLS to share hardware resources (such as multipliers) across multiple operations scheduled in different clock cycles. Thus, not all datapath paths are active at every clock. 
\begin{figure}
    \centering
    \includegraphics[width=\linewidth]{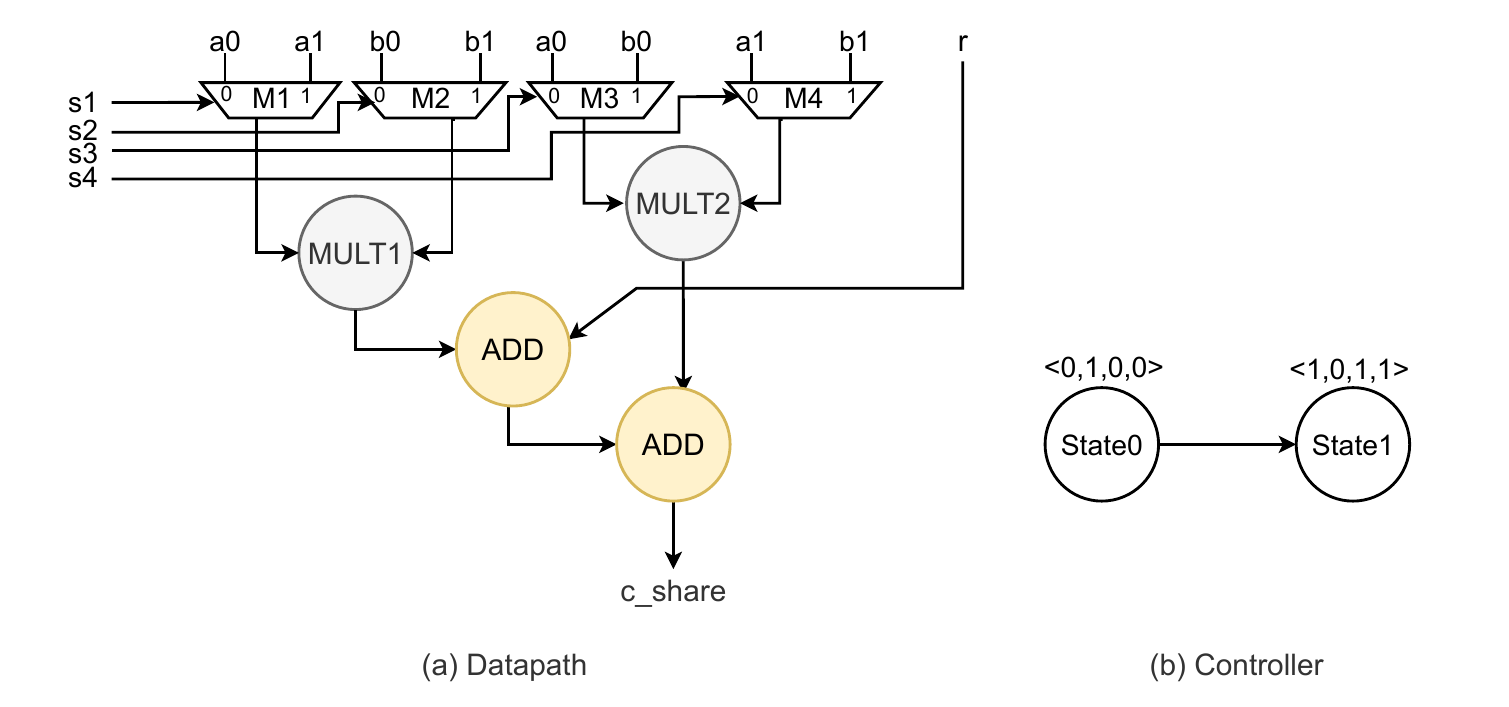}
    \caption{Motivating example for state-based verification of masked hardware}
    \label{fig:motivation0}
\end{figure}
Consider the masked multiplier in Figure~\ref{fig:motivation0}, synthesized with a resource constraint of two multipliers. The four multiplications required by the masking scheme are distributed across two cycles, and multiplexers supply the appropriate input pairs in each state. For example, the valid input pairs for the multiplier $MULT2$ are only
\((a_0, b_0)\) and \((a_1, b_1)\) as these do not need to be masked subsequently after combination as done for the outputs of the multiplier $MULT1$. This is required for secure masking.
Combinations such as \((a_0, b_1)\), although structurally possible, therefore, \emph{never occur} in the actual design because the FSM controls the MUX select lines to never select those inputs by design.

However, \textsc{REBECCA} does not account for the FSM and still evaluates the multiplier under all theoretically possible MUX inputs, including unreachable ones. Consequently, it also evaluates the security of the datapath obtained for \((a_0, b_1)\) for the first multiplier, thus registering a leakage which is a false positive. Thus, REBECCA may report a first-order leakage even when the hardware is secure.

\textbf{Underlying limitation of existing tools: }
These observations reveal a fundamental limitation: \textsc{REBECCA} performs a naive combinational analysis that checks leakage under all syntactically possible MUX inputs, without considering whether these combinations are actually reachable in a given FSM state. This results in false positives for FSM-based, time-division multiplexed datapaths. Table~\ref{tab:motivation} further illustrates this behavior: whenever the RTL includes an FSM (and the corresponding multiplexers), \textsc{REBECCA} produces false positives, whereas designs without FSMs are verified correctly.

\medskip
\noindent\textbf{Our Approach.}
In this work, we propose a \emph{state-aware verification methodology} that leverages the structure of the FSM to:
\begin{enumerate}
    \item identify the portion of the datapath active in each state,
    \item enumerate only the valid input combinations reachable in that state, and
    \item verify masking security state-by-state, ensuring both local and global correctness.
\end{enumerate}
By restricting verification to reachable configurations, our method achieves accurate and efficient security evaluation for HLS-generated masked designs.

{\scriptsize
\begin{table}[htbp]
\centering
\caption{Results obtained from Rebecca}
\label{tab:motivation}
\renewcommand{\arraystretch}{1.2}

\begin{adjustbox}{max width=\columnwidth}
\begin{tabular}{l cc cc cc}
\toprule
\multirow{2}{*}{\textbf{Example}} 
& \multirow{2}{*}{\textbf{FSM}} 
& \multirow{2}{*}{\textbf{Reg}} 
& \multicolumn{2}{c}{\textbf{Expected Op}} 
& \multicolumn{2}{c}{\textbf{Rebecca Op}} \\
\cmidrule(lr){4-5}
\cmidrule(lr){6-7}
& & & $s$ & $t$ & $s$ & $t$ \\
\midrule
AND Gate Circuit     & \xmark & \xmark & True  & True  & True  & True  \\
Circuit with MUX     & \xmark & \xmark & False & False & False & False \\
DOM RTL version1     & \xmark & \xmark & True  & False & True  & False \\
DOM RTL version2     & \xmark & \cmark & True  & True  & True  & True  \\
DOM RTL version3     & \cmark & \xmark & True  & False & \textcolor{red}{False} & \textcolor{red}{False} \\
DOM RTL version4     & \cmark & \cmark & True  & True  & \textcolor{red}{False} & \textcolor{red}{False} \\
\midrule
COMAR RTL            & \cmark & \cmark & True  & True  & \textcolor{red}{False} & \textcolor{red}{False} \\
HPC1 RTL             & \cmark & \cmark & True  & True  & \textcolor{red}{False} & \textcolor{red}{False} \\
HPC2 RTL             & \cmark & \cmark & True  & True  & \textcolor{red}{False} & \textcolor{red}{False} \\
PRESENT DOM RTL      & \cmark & \cmark & True  & True  & \textcolor{red}{False} & \textcolor{red}{False} \\
PRESENT HPC1 RTL     & \cmark & \cmark & True  & True  & \textcolor{red}{False} & \textcolor{red}{False} \\
\bottomrule
\end{tabular}
\end{adjustbox}
\end{table}

}

% \section{Motivation}
% \input{sections/motivation_rewritten}
% \section{Background and Related Works}
\section{background}
\label{sec:background}
This section presents a brief background required to understand the contributions of this article. It also provides a background on REBECCA \cite{rebeccabloem2018formal}, a formal verification tool for masked hardware.
\subsection{Masking}
\label{subsec:masking}
Given a cryptographic algorithm, masking it for a particular order $d - 1$, involves splitting the secret inputs into $d$ shares. Thereafter the computations in the algorithmic description are replaced by their masked versions, i.e., the entire computation-tree is split into $d$ shares. In doing so non-linear functions need to be carefully handled to ensure independence of their output shares (and hence the input shares to the next dependent function). Typically, masking a non-linear function requires \textit{mask refreshing} operations that require extra register stages (that account for glithces in hardware that causes recombination of shares) and fresh random bits and \textit{re-masking} operations that make use of fresh random bits to maintain the independence of shares at the output. This can be done in several different ways each with a different area/latency/number of random bits used for \textit{re-masking}. Coming up with an optimal and correct masked version of a program, that is glitch-robust, is difficult \cite{ThresholdImplementations2006}. Hence, \textit{gadgets} were introduced to mask cryptographic hardware in a modular way. This also aided the scalable verification of gadgets with certain properties could be composed with another verified gadget without compromising their SCA resistance.  Several gadgets like the Domain-Oriented Masking (DOM) \cite{dom2016} gadget, the Hardware Private Circuits version $1$  (HPC1),  and $2$ (HPC2) \cite{HPC2020} gadgets and the COMAR \cite{comar} gadgets exists. These have been used to obtain the masked benchmarks in this paper. 
% define static and transient
\subsection{The Extended Probing Model}
\label{subsec:probing model}
Masking works under the assumption that a real power side channel attacker can be approximated to an attacker, the strength of which is determined by how many intermediate wires (intermediate computation results in software) can be statistically infer from the overall power consumption of the device alone. This is called classic probing model proposed by Ishai et. al., \cite{isw} and is referred to as the static state (S) of a design throughout this paper. 

The static state, however, does not take into account the effects of glitches and transitions on the masking security of the hardware \cite{glitch2005mangard}. Hence, the extended probing model was introduced \cite{robustprobingmodel2018} which assumed that, under a glitchy scenario where all inputs at each gate arrive at different times, the output of each wire would leak the outputs of all previous gates in its fan-in, till a register stage (last synchronization point) is encountered or till the input of the design. This captures the information available to an attacker for all possible combination of arrival times of inputs in all gates of a design. This behavior is referred to as the transient state (T) of a design in the rest of the paper.
%define gadget based masking

%define how sota formal verificaiton tools work
\subsection{Formal Verification of Masked Hardware}
\label{subsec:sota}
Several formal verification tools have been proposed to analyze the probing security of masked hardware implementations. REBECCA~\cite{rebeccabloem2018formal} introduced a SAT-based verification approach that tracks the support of Walsh/Fourier expansions of intermediate signals to detect secret-dependent leakages under static and transient probing models. CoCo~\cite{gigerl2021coco} extends this line of work by enabling compositional verification of masked circuits, improving scalability while still assuming datapath only designs. VERICA~\cite{verica} employs symbolic reasoning to verify masking correctness at the gate level.  While these state-of-the-art tools are effective for handwritten and pipelined masked hardware, they implicitly assume fixed datapaths and over-approximate all possible executions in a design with a controller and a datapth. \textbf{Note: }We use REBECCA as the underlying verification tool because it was developed specifically  for formal analysis of masked hardware circuits under the robust probing model, and takes as input hardware designs at RTL. Tools such as COCO \cite{gigerl2021coco} were originally designed for co-verification of masked software on processor netlists. Tools like VERICA \cite{verica} support combined  side-channel and fault analysis but do not directly target scalable verification of masked RTL as needed in our flow. 

% Consequently, when applied to HLS-generated designs featuring FSM-controlled, resource-shared datapaths, these tools may report false positives by verifying infeasible datapath combinations.
\subsection{REBECCA}
\label{subsec:rebecca}
Since, we have used the REBECCA \cite{rebeccabloem2018formal} formal verification tool to present our findings in this paper, this section provides a brief introduction into the working of REBECCA.
\subsubsection{REBECCA overall flow}
\label{subsubsec:rebeccaflow}
   REBECCA takes as input a masked Verilog design, and a label file (.txt) that contains the labelling information of each input signal as a share, random, or public value.  Additionally, it takes as input the masking security order ($d$) and the leakage model (static (S) and transient (T)). Internally, REBECCA invokes Yosys to synthesize the design and generate a JSON-based structural description containing gate-wise information. From this representation, REBECCA extracts the list of gates $G$, the set of wires $W$, and the Boolean function $f(g)$ computed by each gate $g \in G$. The circuit is modeled as a directed acyclic graph (DAG), where vertices correspond to gates and edges correspond to signal wires, and all primary inputs are classified into disjoint sets of secret variables $S$, mask variables $M$, and public variables $P$ obtained from the label file. Based on these, REBECCA performs formal verification to assess whether the design is secure under the specified probing model (S or T). In order to do so a sound but conservative estimation of the Fourier coefficient of each gate in the design is done. This reflects statistical dependence of the output of each gate on each input of the design and thus isolates secret dependent computations and thus possible leakages. This also enables localization of the flaws detected by the verifier. REBECCA leverages SAT-based formal verification making it faster than previous methods. The output of REBECCA is a log file with which indicates if the design is secure or not and if not then generates the location of the violation, helping the designer to debug the masked hardware. 
   % Figure \ref{fig:rebeccaflow} presents an overall flow of the REBECCA verification tool, the steps of which are detailed in this subsection.
   
        % \begin{figure}
        %     \centering
        %     \includegraphics[width=0.8\linewidth]{figures/backround/REBECCA.drawio.pdf}
        %     \caption{REBECCA overall flow}
        %     \label{fig:rebeccaflow}
        % \end{figure}
        
\subsubsection{Fourier Expansion of a Boolean Function}
\label{subsubsec:fourier}
% why fourier coefficeint estimation is inaccurate for security analysis
Underlying REBECCA's verification engine is the representation of Boolean signals in the $\{\pm1\}$ domain rather than the conventional $\{0,1\}$ domain. This transformation enables an algebraic structure that makes Fourier analysis of Boolean functions scalable for leakage assessment via REBECCA. When Boolean values are mapped via $0\mapsto +1$ and $1\mapsto -1$, the XOR operation becomes:
\[
a \oplus b = (a + b) \bmod 2
\]

This expression is hard to analyze algebraically as its parity is nonlinear. However, in the \(\{+1,-1\}\) domain:
\[
(-1)^a \cdot (-1)^b = (-1)^{a \oplus b}
\]
So XOR becomes multiplication:
\[
\overline{a \oplus b} \;=\; \bar{a} \cdot \bar{b}
\]
While negation becomes multiplication by $-1$, Boolean AND becomes a polynomial making it easier for Fourier analysis. 

A Fourier expansion expresses any two-input Boolean function $f(x, y)$ ($x, y \in X$; $X \in \{\pm{1}\}$ is the domain of $f$) as a linear combination of the basis functions $\{1, x, y, xy\}$. Accordingly, the AND function can be written in the form $\tilde f(x,y)=a+bx+cy+dxy$, where the coefficients are determined by matching the function values to the AND truth table. Solving the resulting linear system gives the values of $a, b, c,$ and $d$. Thus, the Fourier expansion of $AND(x)$ is given by $\tfrac{1}{2} + \tfrac{1}{2}x + \tfrac{1}{2}y - \tfrac{1}{2}xy$. Each of the \textit{basis functions} corresponds to a distinct dependency pattern: constant, single-variable or joint, allowing REBECCA to track the presence of secret-dependent correlations by reasoning solely about which basis components appear in a gate’s expansion. In the context of masking violation, this means: a non-zero Fourier coefficient of a monomial in the expansion indicates statistical correlation between the gate output and the XOR of the variables making up the monomial. This forms the basis of REBECCA: \textit{a gate output is masking-insecure when it correlates with a subset of input values that contains secret variables but no mask variables.}

\subsubsection{Label in REBECCA}
\label{subsubsec:labelling}
In REBECCA, a \textit{label} of a gate represents the set of monomials over the input variables whose coefficients are non-zero in the Fourier expansion Boolean function computed by a gate. Rather than tracking the exact coefficient values, labels conservatively capture which combinations of inputs may influence a signal. Given the labels of the input to a gate, the label of its output can be derived via the gate-specific rules of labeling in Table \ref{tab:rules}. They can be derived from the Fourier expansion of the AND and XOR operations following the discussion in Section \ref{subsubsec:fourier}.

Given two inputs of a gate labeled as (sets of monomials), $S_a$ and $S_b$, REBECCA defines abstract operators $\oplus$ and $\otimes$. The operator $\oplus$ corresponds to cancellation of identical monomials, i.e., $S_a \oplus S_b = (S_a \cup S_b) \setminus (S_a \cap S_b)$. The $\otimes$ operator models monomial multiplication and is defined as $S_a \otimes S_b = \{ m_a m_b \mid m_a \in S_a,\, m_b \in S_b \}$ capturing the joint-distribution due to a \textit{non-linear} AND operation for instance. In the presence of glitches, it follows from the extended probing model that \cite{robustprobingmodel2018} a linear gate may temporarily output intermediate values (which algebrically corresponds to a non-linear gate's behavior) due to the mismatched input arrival times of its inputs. This is conservatively modeled by REBECCA by labeling a linear gate's output using rules for a non-linear gate under the transient case. The rules for MUX can be derived form the rules for the AND, NOT and XOR operations. 

\textbf{Example: Labelling in REBECCA}
Given the circuit in Figure \ref{fig:label_propagation}a, the inputs are labeled as shown in Table \ref{tab:example} where $s_m$ denotes a secret masked with a value $m$, $m_s$ denotes a mask which combined with the masked secret $s_m$ would leak the secret $s$.  
The labeling proceeds from the input towards the output such that each gate's labels are derived from the labels of its inputs. For the given circuit, the inputs are labeled as shown in Table \ref{tab:example}.
Thus, for the gate $G_1$, a XOR gate, the labels are obtained as $S_{g1} = \{s, m_s\} \oplus \{m_1\} = \{\{s, m_1, m_s\}\}$. However, $T_{g1} = \{s, m_s\} \otimes \{m_1\} = \{\phi, \{s, m_s\}, \{m_1\},\{s, m_1, m_s\}\}$.
For the gate $G_2$, a AND gate, the labels are obtained as $S_{g2} = \{m_s\} \otimes \{p_1\} = \{\phi, \{m_s\}, \{p_1\},\{m_s, p_1\}\}$. Similarly, $T_{g2} = \{m_s\} \otimes \{p_1\} = \{\phi, \{m_s\}, \{p_1\},\{m_s, p_1\}\}$.
Similarly, the labelling of $G_{3}$ are obtained using the rules in Table \ref{tab:rules} for the static (non-glitchy) and transient (glitchy) scenarios as shown in Figure \ref{fig:label_propagation}e and Figure \ref{fig:label_propagation}f respectively. The blue colored labels in Figure \ref{fig:label_propagation}f indicate labels that are in common with the labels in Figure \ref{fig:label_propagation}e. The red colored labels in Figure \ref{fig:label_propagation}f are the labels that are generated specifically for the transient case. Given such a labeling, REBECCA checks each gate's output for correlations with secrets. If they do exist, REBECCA further checks if they contain a correlation to a mask unrelated to the secret. If such a mask exists in the correlation set then the gate is secure otherwise it is not secure. In this example, for the static case there are no such labels with only a secret (or a secret with public values). But for the transient case, the design has two labels (underlined) : $\{s\}$ and $\{s, p_1\}$ at the output of $G_3$ indicating that the design is secure in the static probing model and insecure in the transient probing model, with a leak at $G_3$.

\begin{figure*}[]
\centering
\begin{minipage}[t]{0.45\textwidth}
\begin{tikzpicture}[node distance=1cm, font=\sffamily, scale=0.6, transform shape]
    % Nodes
    \node (input1) {};
    \node (input2) [below of=input1, yshift=-0.5cm] {};
    \node (input3) [below of=input2, yshift=-0.2cm] {};
    \node (input4) [below of=input3, yshift=-0.5cm] {};
    \node (xor1) [right of=input1, xshift=1.5cm, yshift=-0.5cm, draw, xor gate US, logic gate inputs=nn, anchor=input 1] {XOR};
    \node (and) [right of=input3, xshift=1.5cm, draw, and gate US, logic gate inputs=nn, anchor=input 1] {AND};
    \node (xor2) [right of=and, xshift=0.7cm, yshift=1.4cm, draw, xor gate US, logic gate inputs=nn, anchor=input 1] {XOR};
    \node (output) [right of=xor2, xshift=2cm] {};
   
    % Arrows
    \draw[->] (input1.east) -- ++(0.3,0) |- (xor1.input 1) node[pos=0.75, above] {${S_{m}}$};
    \draw[->] (input2.east) -- ++(0.3,0) |- (xor1.input 2) node[pos=0.75, below] {${m_{1}}$};
    \draw[->] (xor1.output) -- ++(0.3,0) |- (xor2.input 1);
    \draw[->] (input3.east) -- ++(0.3,0) |- (and.input 1) node[pos=0.75, above] {${m_{s}}$};
    \draw[->] (input4.east) -- ++(0.3,0) |- (and.input 2) node[pos=0.75, below] {${p_{1}}$};
    \draw[->] (and.output) -- ++(0.3,0) |- (xor2.input 2);
    \draw[->] (xor2.output) -- (output.west) node[pos=0.35, below] {$q$};
    
    % Labels
    \node [below of=xor1, yshift=-0.5cm] {${G_1}$};
    \node [below of=and, yshift=-0.5cm] {${G_2}$};
    \node [below of=xor2, yshift=-0.5cm] {${G_3}$};
    
    % Caption
    \node [below of=and, yshift=-1cm] {(a) Circuit to be verified};
\end{tikzpicture}
\end{minipage}
\hfill
\begin{minipage}[t]{0.45\textwidth}
\begin{tikzpicture}[node distance=1cm, font=\sffamily,scale=0.6, transform shape]
    % Nodes
    \node (input1) [draw, rectangle, text width=1.5cm, align=center] {{$\{ s, m_s \}$}};
    \node (input2) [below of=input1, yshift=-0.5cm, draw, rectangle, text width=1.5cm, align=center] {{$\{ m_1 \}$}};
  
    \node (xor1) [right of=input1, xshift=1.5cm, yshift=-0.5cm, draw, xor gate US, logic gate inputs=nn, anchor=input 1] {XOR};
    
    \node (wire2)[right of=xor1, xshift=1.4cm, draw, rectangle, text width=2.6cm, align=center] {{$\{\{ s,  m_s, m_1\}\}$}};
     
    % Arrows
    \draw[->] (input1.east) -- ++(0.3,0) |- (xor1.input 1)node[pos=0.75, above] {${S_{m}}$};
    \draw[->] (input2.east) -- ++(0.3,0) |- (xor1.input 2)node[pos=0.75, below] {${m_{1}}$};
     \draw[->] (xor1.east) -- ++(0.2,0) |- (wire2);
    % \draw[->] (xor1.output) -- ++(0.3,0) |- (xor2.input 1);

    % Labels
    % \node [above of=input1, yshift=-0.3cm] {${S_m}$};
    \node [below of=xor1] {${G_1}$};
    
    % \node [right of=output, xshift=1cm] {${q}$};
    
    % Gate label
    \node [below of=xor1, yshift=-1.1cm] {(b) Stable Propagation of labels for the first XOR gate};
\end{tikzpicture}
\end{minipage}
~
\begin{minipage}[t]{0.45\textwidth}
\begin{tikzpicture}[node distance=1cm, font=\sffamily,scale=0.6, transform shape]
    % Nodes
    
    \node (input3) [below of=input2, yshift=-0.2cm, draw, rectangle, text width=1.5cm, align=center] {{$\{ m_s \}$}};
    \node (input4) [below of=input3, yshift=-0.5cm, draw, rectangle, text width=1.5cm, align=center] {{$\{ p_1 \}$}};
  
    \node (and) [right of=input3, xshift=1.5cm, draw, and gate US, logic gate inputs=nn, anchor=input 1] {AND};

    \node (wire2)[right of=and, xshift=2.4cm,  draw, rectangle, text width=3.8cm, align=center] {{$\{ \phi,  \{m_s\}, \{ p_1\}, \{ m_s, p_1\}\}$}};
    % Arrows
    
    \draw[->] (input3.east) -- ++(0.3,0) |- (and.input 1)node[pos=0.75, above] {${m_{s}}$};
    \draw[->] (input4.east) -- ++(0.3,0) |- (and.input 2)node[pos=0.75, below] {${p_{1}}$};
    \draw[->] (and.east) -- ++(0.2,0) |- (wire2);

    % Labels
    % \node [above of=input1, yshift=-0.3cm] {${S_m}$};

    \node [below of=and] {${G_2}$};

    % \node [right of=output, xshift=1cm] {${q}$};
    
    % Gate label
    \node [below of=and, yshift=-1cm] {(c) Stable and Transient Propagation of Labels for the AND gate};
\end{tikzpicture}
\end{minipage}
\hfill
\begin{minipage}[t]{0.45\textwidth}
\begin{tikzpicture}[node distance=1cm, font=\sffamily,scale=0.6, transform shape]
    % Nodes
   
    \node (xor2) [right of=and, xshift=0.7cm, yshift=1.4cm, draw, xor gate US, logic gate inputs=nn, anchor=input 1] {XOR};
    \node (output) [right of=xor2, xshift=2cm, draw, rectangle, text width=3.2cm, align=center] {{$\{ \{s, m_s, m_1\}, \{s, m_1\}$,\\ $\{s, m_s, m_1, p_1\}, $\\$\{s, m_1, p_1\} \}$}};
    \node (wire2)[right of=xor1,  xshift=-1.7cm, draw, rectangle, text width=2.6cm, align=center] {{$\{\{ s,  m_s, m_1\}\}$}};
    \node (wire2)[right of=and, xshift=-2.4cm, draw, rectangle, text width=3.8cm, align=center] {{$\{ \phi,  \{m_s\}, \{ p_1\}, \{ m_s, p_1\}\}$}};
    % Arrows
   
    \draw[->] (xor1.output) -- ++(0.3,0) |- (xor2.input 1);
    
    \draw[->] (and.output) -- ++(0.3,0) |- (xor2.input 2);
    \draw[->] (xor2.output) -- (output.west)node[pos=0.35, below] {$q$};
    
    % Labels
    % \node [above of=input1, yshift=-0.3cm] {${S_m}$};

    \node [below of=xor2] {${G_3}$};
    % \node [right of=output, xshift=1cm] {${q}$};
    
    % Gate label
    \node [below of=xor2, yshift=-1cm] {(d) Stable Propagation of Labels for the Second XOR gate};
\end{tikzpicture}
\end{minipage}
~
\begin{minipage}[t]{0.45\textwidth}
\begin{tikzpicture}[node distance=1cm, font=\sffamily,scale=0.6, transform shape]
    % Nodes
    \node (input1) [draw, rectangle, text width=1.5cm, align=center] {{$\{ s, m_s \}$}};
    \node (input2) [below of=input1, yshift=-0.5cm, draw, rectangle, text width=1.5cm, align=center] {{$\{ m_1 \}$}};
    \node (input3) [below of=input2, yshift=-0.2cm, draw, rectangle, text width=1.5cm, align=center] {{$\{ m_s \}$}};
    \node (input4) [below of=input3, yshift=-0.5cm, draw, rectangle, text width=1.5cm, align=center] {{$\{ p_1 \}$}};
    \node (xor1) [right of=input1, xshift=1.5cm, yshift=-0.5cm, draw, xor gate US, logic gate inputs=nn, anchor=input 1] {XOR};
    \node (and) [right of=input3, xshift=1.5cm, draw, and gate US, logic gate inputs=nn, anchor=input 1] {AND};
    \node (xor2) [right of=and, xshift=0.7cm, yshift=1.4cm, draw, xor gate US, logic gate inputs=nn, anchor=input 1] {XOR};
    \node (output) [right of=xor2, xshift=2cm, draw, rectangle, text width=3.2cm, align=center] {{$\{ \{s, m_s, m_1\}, \{s, m_1\}$,\\ $\{s, m_s, m_1, p_1\}, $\\$\{s, m_1, p_1\} \}$}};
    \node (wire2)[right of=xor1, xshift=1.4cm, yshift=0.5cm, draw, rectangle, text width=2.6cm, align=center] {{$\{\{ s,  m_s, m_1\}\}$}};
     \node (wire2)[right of=and, xshift=1.4cm, yshift=-0.8cm, draw, rectangle, text width=3.8cm, align=center] {{$\{ \phi,  \{m_s\}, \{ p_1\}, \{ m_s, p_1\}\}$}};
    % Arrows
    \draw[->] (input1.east) -- ++(0.3,0) |- (xor1.input 1)node[pos=0.75, above] {${S_{m}}$};
    \draw[->] (input2.east) -- ++(0.3,0) |- (xor1.input 2)node[pos=0.75, below] {${m_{1}}$};
    \draw[->] (xor1.output) -- ++(0.3,0) |- (xor2.input 1);
    \draw[->] (input3.east) -- ++(0.3,0) |- (and.input 1)node[pos=0.75, above] {${m_{s}}$};
    \draw[->] (input4.east) -- ++(0.3,0) |- (and.input 2)node[pos=0.75, below] {${p_{1}}$};
    \draw[->] (and.output) -- ++(0.3,0) |- (xor2.input 2);
    \draw[->] (xor2.output) -- (output.west)node[pos=0.35, below] {$q$};
    
    % Labels
    % \node [above of=input1, yshift=-0.3cm] {${S_m}$};
    \node [below of=xor1] {${G_1}$};
    \node [below of=and] {${G_2}$};
    \node [below of=xor2] {${G_3}$};
    % \node [right of=output, xshift=1cm] {${q}$};
    
    % Gate label
    \node [below of=and, yshift=-1cm] {(e) Using Stable Propagation of Labels};
\end{tikzpicture}
\end{minipage}
\hfill
\begin{minipage}[t]{0.45\textwidth}
\begin{tikzpicture}[node distance=1cm, font=\sffamily,scale=0.6, transform shape]
    % Nodes
    \node (input1) [draw, rectangle, text width=1.5cm, align=center] {\textcolor{blue}{$\{ s, m_s \}$}};
    \node (input2) [below of=input1, yshift=-0.5cm, draw, rectangle, text width=1.5cm, align=center] {\textcolor{blue}{$\{ m_1 \}$}};
    \node (input3) [below of=input2, yshift=-0.2cm, draw, rectangle, text width=1.5cm, align=center] {\textcolor{blue}{$\{ m_s \}$}};
    \node (input4) [below of=input3, yshift=-0.5cm, draw, rectangle, text width=1.5cm, align=center] {\textcolor{blue}{$\{ p_1 \}$}};
    \node (xor1) [right of=input1, xshift=1.5cm, yshift=-0.5cm, draw, xor gate US, logic gate inputs=nn, anchor=input 1] {XOR};
    \node (and) [right of=input3, xshift=1.5cm, draw, and gate US, logic gate inputs=nn, anchor=input 1] {AND};
    \node (xor2) [right of=and, xshift=1.7cm, yshift=1.4cm, draw, xor gate US, logic gate inputs=nn, anchor=input 1] {XOR};
    \node (output) [right of=xor2, xshift=2.3cm, draw, rectangle, text width=4.2cm, align=center] {{$\{\textcolor{red}{\{s, m_{s}\}, \textcolor{blue}{\{m_1\},}}$\\$\textcolor{blue}{ \{s, m_s, m_1\}, \{s, m_1\}}$,\\$\textcolor{red}{\{m_s\}, \{ p_1\}, \{ m_s, p_1\}}$\\ $\textcolor{blue}{\{s, m_s, m_1, p_1\}}, \underline{\textcolor{red}{\{s\}}},$\\$\textcolor{red}{\{m_1\}},\textcolor{red}{\{s, m_s, p_1\}, \{ m_1, p_1\},}$\\$\underline{\textcolor{red}{\{s, p_1\}}}, \textcolor{blue}{\{s, m_1, p_1\}} \}$}};
    \node (wire1)[right of=xor1, xshift=1.4cm, yshift=0.7cm, draw, rectangle, text width=2.6cm, align=center] {{$\{\textcolor{red}{\{s, m_{s}\}, \{m_1\},}$\\ $\textcolor{blue}{\{ s,  m_s, m_1\}}\}$}};
    \node (wire2)[right of=and, xshift=1.4cm, yshift=-0.8cm, draw, rectangle, text width=3.8cm, align=center] {\textcolor{blue}{$\{ \phi,  \{m_s\}, \{ p_1\}, \{ m_s, p_1\}\}$}};
    % Arrows
    \draw[->] (input1.east) -- ++(0.3,0) |- (xor1.input 1)node[pos=0.75, above] {${S_{m}}$};
    \draw[->] (input2.east) -- ++(0.3,0) |- (xor1.input 2)node[pos=0.75, below] {${m_{1}}$};
    \draw[->] (xor1.output) -- ++(0.3,0) |- (xor2.input 1);
    \draw[->] (input3.east) -- ++(0.3,0) |- (and.input 1)node[pos=0.75, above] {${m_{s}}$};
    \draw[->] (input4.east) -- ++(0.3,0) |- (and.input 2)node[pos=0.75, below] {${p_{1}}$};
    \draw[->] (and.output) -- ++(0.3,0) |- (xor2.input 2);
    \draw[->] (xor2.output) -- (output.west)node[pos=0.35, below] {$q$};
    
    % Labels
    % \node [above of=input1, yshift=-0.3cm] {${S_m}$};
    \node [below of=xor1] {${G_1}$};
    \node [below of=and] {${G_2}$};
    \node [below of=xor2] {${G_3}$};
    % \node [right of=output, xshift=1cm] {${q}$};
    
    % Gate label
    \node [below of=and, yshift=-1cm] {(f) Using Transient Propagation of Labels};
\end{tikzpicture}
\end{minipage}
\caption{Label propagation through logic gates in the verification circuit.}
\label{fig:label_propagation}
\end{figure*}
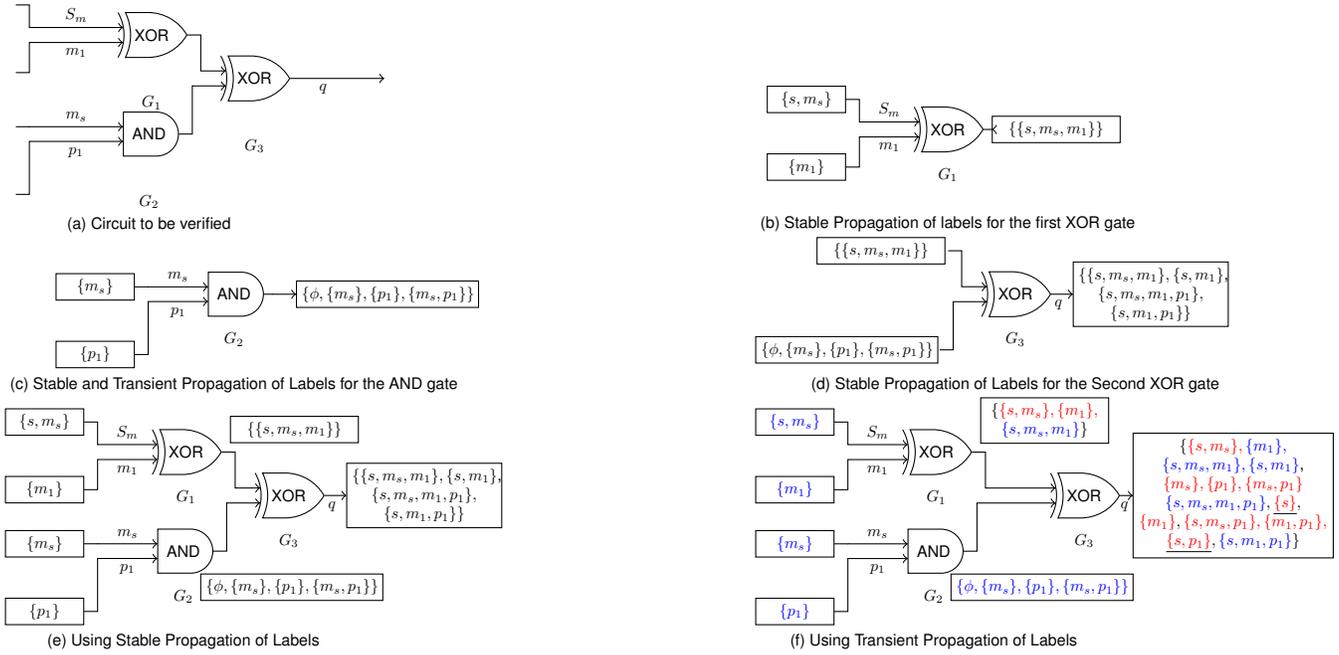

\begin{table}[H]
\centering
\caption{Rules for gates in stable and transient states}
\label{tab:rules}
\renewcommand{\arraystretch}{1.2}

\begin{adjustbox}{max width=\columnwidth}
\begin{tabular}{l c c c}
\toprule
 & \textbf{Gate} & \textbf{Stable} & \textbf{Transient} \\
\midrule
Wire &
\begin{tikzpicture}[scale=0.7]
    \draw (0,0) node[left] {$a$} -- (1,0) node[right] {$a$};
\end{tikzpicture}
& $\{a\}$ & $\{a\}$ \\\\
\midrule
NOT &
\begin{tikzpicture}[circuit logic US, scale=0.6]
    \node[not gate, draw, logic gate inputs=none] (NOT1) at (0,0) {};
    \draw (NOT1.input) -- ++(-0.4,0) node[left] {$a$};
    \draw (NOT1.output) -- ++(0.4,0) node[right] {$a^\prime$};
\end{tikzpicture}
& $\{a\}$ & $\{a\}$ \\\\
\midrule
AND / NAND &
\begin{tikzpicture}[circuit logic US, scale=0.8]
    \node[and gate, draw, logic gate inputs=nn] (AND1) at (0,0) {};
    \draw (AND1.input 1) -- ++(-0.4,0) node[left] {$a$};
    \draw (AND1.input 2) -- ++(-0.4,0) node[left] {$b$};
    \draw (AND1.output) -- ++(0.4,0) node[right] {$c$};
\end{tikzpicture}
& $S_a \otimes S_b$ & $T_a \otimes T_b$ \\\\
\midrule
Register &
\begin{tikzpicture}[scale=0.6]
    \draw (0,0) node[left] {$a$} -- (0.4,0);
    \draw (0.4,0.25) rectangle (1,-0.25);
    \draw (1,0) -- (1.4,0) node[right] {$a$};
\end{tikzpicture}
& $\{a\}$ & $\{a\}$ \\\\
\midrule
XOR / XNOR &
\begin{tikzpicture}[circuit logic US, scale=0.8]
    \node[xor gate, draw, logic gate inputs=nn] (XOR1) at (0,0) {};
    \draw (XOR1.input 1) -- ++(-0.4,0) node[left] {$a$};
    \draw (XOR1.input 2) -- ++(-0.4,0) node[left] {$b$};
    \draw (XOR1.output) -- ++(0.4,0) node[right] {$c$};
\end{tikzpicture}
& $S_a \oplus S_b$ & $T_a \otimes T_b$ \\\\
\midrule
MUX &
\begin{tikzpicture}[scale=0.5]
    \draw (0,1) -- (1,0.8) -- (1,-0.8) -- (0,-1) -- cycle;
    \draw (-0.8,0.5) node[left] {$a$} -- (0,0.5);
    \draw (-0.8,-0.5) node[left] {$b$} -- (0,-0.5);
    \draw (1,0) -- (1.8,0) node[right] {$y$};
    \draw (0.5,-1) -- (0.5,-1.2) -- (-0.8,-1.2) node[left] {$c$};
    \node at (0.5,0) {\tiny{MUX}};
\end{tikzpicture}
& $S_c \otimes (S_a \cup S_b)$ & $T_a \otimes T_b \otimes T_c$ \\\\
\bottomrule
\end{tabular}
\end{adjustbox}
\end{table}
\begin{table}[H]
\centering
\caption{Input labeling for the design in the example shown in Fig.~\ref{fig:label_propagation}}
\label{tab:example}
\renewcommand{\arraystretch}{1.2}

\begin{adjustbox}{max width=\columnwidth}
\begin{tabular}{c l l l}
\toprule
\textbf{No.} & \textbf{Inputs} & \textbf{Label} & \textbf{Correlation Set} \\
\midrule
1 & $s_m$ & Share 1 & $\{s_1, m_2\}$ \\
2 & $m_s$ & Share 1 & $\{m_2\}$ \\
3 & $m_1$ & Mask & $\{m_1\}$ \\
4 & $p_1$ & Unimportant (public value) & $\{p_1\}$ \\
5 & $q$ & Unimportant (public value) & $\{q\}$ \\
% 6 & XOR ($G_1$) & -- & -- \\
% 7 & AND ($G_2$) & -- & -- \\
% 8 & XOR ($G_3$) & -- & -- \\
\bottomrule
\end{tabular}
\end{adjustbox}
\end{table}
\subsubsection{Complexity of REBECCA}
The label-propagation steps takes linear time with the number of gates in the input, as each gate is analyzed once. The detection of masking violations is done via a set of SAT queries whose size grows polynomially with the number of gates and wires in the design. Thus, the computational complexity of REBECCA is bound by the SAT-based verification step. 

\subsection{Using High-Level Synthesis to generate Masked Hardware}
\label{subsec:MaskedHLS}
HLS has been used to generated masked hardware automatically from masked software due to the shortened design time \cite{shortestpath2021}, \cite{hlsarx}, \cite{maskedhls}. HLS obtains a (often optimized) hardware design at RTL from a software specification via the following steps: scheduling, resource-allocation and binding, datapath and controller generation. For any arbitrary masking scheme, it is non-trivial to obtain a hardware masked design from a software masked design. However, for gadget-based masked designs, the process becomes simpler. To obtain a gadget-based masked hardware from a gadget-based masked software of the same algorithm via HLS, the following steps are performed: 
\begin{itemize}
    \item For each gadget in the gadget-based masked software, identify the locations where registers need to be inserted according to the masking scheme for glitch-robust masking. 
    \item Annotate these locations and all parallel paths such that the HLS tool understands that these computations results are to be stored instead of wires in the output RTL.
    \item Pass this annotated software specification as input to the HLS tool to obtain the corresponding masked RTL design.
\end{itemize}

% \section{Impact of HLS}
\section{Impact of HLS on masking}
\label{sec:impactofhls}
Although HLS seems convenient for obtaining masked hardware, it comes with certain challenges. Given a C/C++ code, HLS can output an area/latency-optimized RTL code after performing different optimizations. However, it has been observed that these optimizations impact the PSCA security of a masked design \cite{shortestpath2021, maskedhls}. Below are a few examples illustrating this observation.

\subsection{HLS front-end} The front-end of HLS is the C compilation stage. It translates the C/C++ code into an intermediate representation (IR) using a compiler like GCC or LLVM. The optimizations done in this phase are expression simplification, code motion, re-association, etc. We present a few examples of such optimizations below and also show how they impact the side-channel security of the IR.
\lstset{
  language=C,
  basicstyle=\ttfamily\scriptsize,
  keywordstyle=\color{blue},
  commentstyle=\color{black},
  numbers=left,
  numberstyle=\tiny,
  stepnumber=1,
  frame=single,
  breaklines=true,
  tabsize=2,
  numbersep=6pt,
  xleftmargin=2em,
  framexleftmargin=2em
}
 \begin{lstlisting}[ caption={DOMAND expression},captionpos=b, label={lst:expressionbalancing}] 
 #include "ap_int.h"
 ap_int<9> domand (ap_int<9> a0, ap_int<9> a1, 
 ap_int<9> b0, ap_int<9> b1, ap_int<9> z, 
 ap_int<9> *y0, ap_int<9> *y1) {
     p1 = a0 & b0;
     p2 = a0 & b1;
     p3 = a1 & b0;
     p4 = a1 & b1;
     i1 = p2 ^ z;
     i2 = p3 ^ z;
    *y0 = i1 ^ p1;
    *y1 = i2 ^ p4;
 return 0;}
\end{lstlisting}
{\it Reassociation:} The reassociation optimization is related to the LLVM compiler. During compilation, the LLVM compiler may re-associate some of the intermediate operations which results in recombination of shares, while maintaining functional correctness, leading to removal of masking security within the algorithm. This insecurity was identified in the Bambu HLS tool by Sadhukhan et.al.\cite{shortestpath2021}. We present an example C code in the Listing \ref{lst:expressionbalancing} where masking requires the XOR gates $i1$ and $i2$ after the cross-domain products $p2$ and $p3$ for \textit{re-masking}. But, on using LLVM for compilation (as Bambu does in the front end), it shifts the XOR operations and masks the products $p1$ and $p4$ instead. Therefore, $p2$ and $p3$ won't remain \textit{re-masked} and the circuit will become insecure. Specifically, $y0$ in Listing \ref{lst:expressionbalancing} is computed as $y0 = ((a0 \otimes b1) \oplus z) \oplus (a0 \otimes b0)$. Reassociation done  by LLVM changes this computation to $y0 = ((a0 \otimes b0) \oplus z) \oplus (a0 \otimes b1)$ instead. This problem is also not avoidable using pragmas.

\textit{Expression Balancing:} Similarly, Vivado HLS, in its front end GCC compilation, rearranges the operations using associative and commutative properties to construct a balanced computation-tree to reduce the delay of the resultant design. This optimization, known as expression balancing, also results in similar masking security violations. Notably, for certain data types these optimizations are enabled by default and can not be avoided.
 \subsection{HLS back-end}
 In the back end of HLS, the scheduling, resource allocation and binding and datapath and controller generation steps are performed. There are lot more optimizations are applied here by the HLS tool. For example, multiple variables are stored in the same register if their lifetimes are non-overlapping. Two operations can be performed using the same functional units if their schedule are not overlapping. Multiplexers are added to allow the shared resource architecture. There can be multicycle, pipelined and dataflow architecture generated by HLS. Therefore, the optimizations of the backend of HLS may also impact the masking security. This part is mostly not studied properly. 

Thus it is clear that the designer needs to verify the output of HLS-generated masked hardware in order to be sure that HLS-optimizations do not hamper the PSCA-security guaranteed by masking. This work aims to achieve an accurate estimation of the PSCA-security of HLS-generation masked hardware. 
% \section{State-wise Verification Flow}
%%%%%%%%%%%%%%%%%%%%%%%%%%%%%%%%%%%%%%%%%%%%%%%%%%%%%%%%%%%%%%%%%%%%%%%%%%%%%%%%
\section{MaskedHLSVerif: Verification of HLS-generated Masked Designs}
\label{sec:methodology}

Based on the previous discussion, this Section presents our verification flow for HLS-generated masked hardware. 
\subsection{Intuition for State-wise Verification Flow}
\label{subsec:intuition}
HLS tools' scheduling policy introduces a Finite-State Machine (FSM) that asserts the proper control signals as control input to the MUXes, which selects the correct set of inputs at a particular state. A state in the FSM represents the part of the datapath that will be enabled for certain clock cycle, its corresponding operations and the inputs associated. The assignment of the correct control inputs at each state, to the MUXes ensures that an invalid combination does not occur. This characteristic of FSM can be used to develop a verification strategy that will abstract out the design being used in a particular state of the design and verify it state-wise. This will ensure that the design is secure at every state and as a whole when the last state of FSM is reached. This approach will only verify the part of the design that is valid at each state and hence invalid combinations will never occur during verification. In the subsequent subsections, we present the algorithm's for this approach and its steps in detail, starting with the splitting of the original design under verification into state-wise designs for verification. 

\subsection{Splitting into state-wise designs}
\label{subsec:splitting}

Given the design, the verification flow starts with the splitting of the overall design into state-wise operations for obtaining the state-wise designs from the original design obtained via HLS. For each design $D_x$ corresponding to state $state_x$, the following can be defined
\begin{itemize}
    \item Datapath: The datapath corresponding to the operations performed at that state and the datapath of the state-wise design of the dependent instructions from the previous state.
    \item Input signals: Input signals of the datapath of the current state and the input signals of the datapath of all the previous states with dependent operations in the current state.
    \item Output signals: intermediate outputs of the present state operations. (Original output signal of the circuit for the last state)
\end{itemize}
\lstset{
  language=C,
  basicstyle=\ttfamily\scriptsize,
  keywordstyle=\color{blue},
  commentstyle=\color{black},
  numbers=left,
  numberstyle=\tiny,
  stepnumber=1,
  frame=single,
  breaklines=true,
  tabsize=2,
  numbersep=6pt,
  xleftmargin=2em,
  framexleftmargin=2em
}
\begin{lstlisting}[caption={Cascaded masked multiplier design input to HLS},captionpos=b,  label={lst:resourceallocation}] 
#include "ap_int.h"
ap_int<9> multiply (ap_int<9>a0, ap_int<9>a1) {
#pragma HLS INLINE off
return a0 & a1;
}
ap_int<9>top (ap_int<9>a0, ap_int<9>a1, 
ap_int<9>b0, ap_int<9>b1, ap_int<9>r1,
ap_int<9>d0, ap_int<9>d1, ap_int<9>r2,
ap_int<9>*y0, ap_int<9>*y1) {
#pragma HLS EXPRESSION_BALANCE off
#pragma HLS allocation instances=multiply limit=2 
function //above pragma enables resource sharing
ap_int<9> c0, c1;
c0 = (multiply(a0, b1) + r1 ) ^ multiply(a0, b0);
c1 = (multiply(a1, b0) + r1 ) ^ multiply(a1, b1);
*y0 = (multiply(c0, d1) + r2 ) ^ multiply(c0, d0);
*y1 = (multiply(c1, d0) + r2 ) ^ multiply(c1, d1);
return 0;}
 \end{lstlisting}
\begin{figure}
    \centering
    \begin{subfigure}{1\linewidth}
        \centering
        \includegraphics[width=0.8\linewidth]{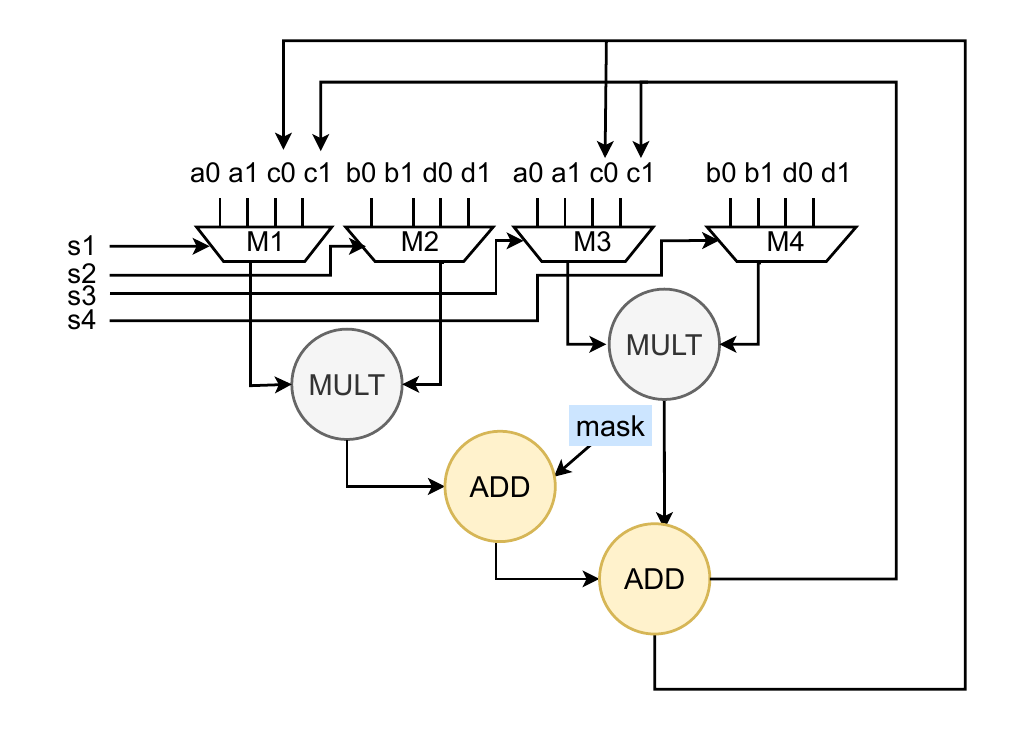}
        \caption{Datapath}
        \label{fig:intuition0-datapath}
    \end{subfigure}
    
    \vspace{0.5cm} % space between the subfigures

    \begin{subfigure}{\linewidth}
        \centering
        \includegraphics[scale=0.5]{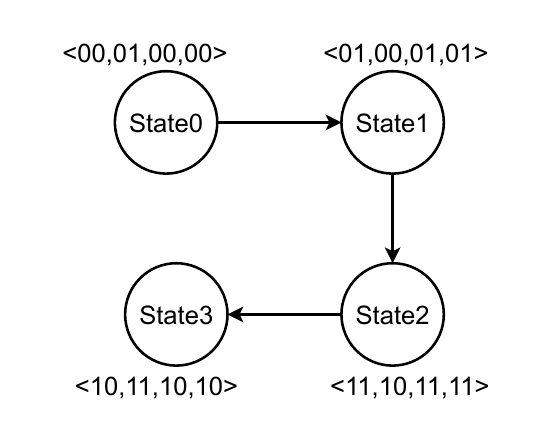}
        \caption{Controller. Here, the control signals are $<s1,s2,s3,s4>$ corresponding to the MUX select lines for MUXes M1, M2, M3, M4 in Figure \ref{fig:intuition0-datapath}. Each $si$, $i \in [1,4]$, is a two-bit value.}
       % \vspace{0.5cm}
        \label{fig:intuition0-controller}
    \end{subfigure}
    
    \caption{HLS-output controller and datapath for the C-design in Listing \ref{lst:resourceallocation} with resource constraints of two MULT per clock cycle.}
    \label{fig:intuition0}
\end{figure}
\begin{figure}
    \centering
    \includegraphics[width=1\linewidth]{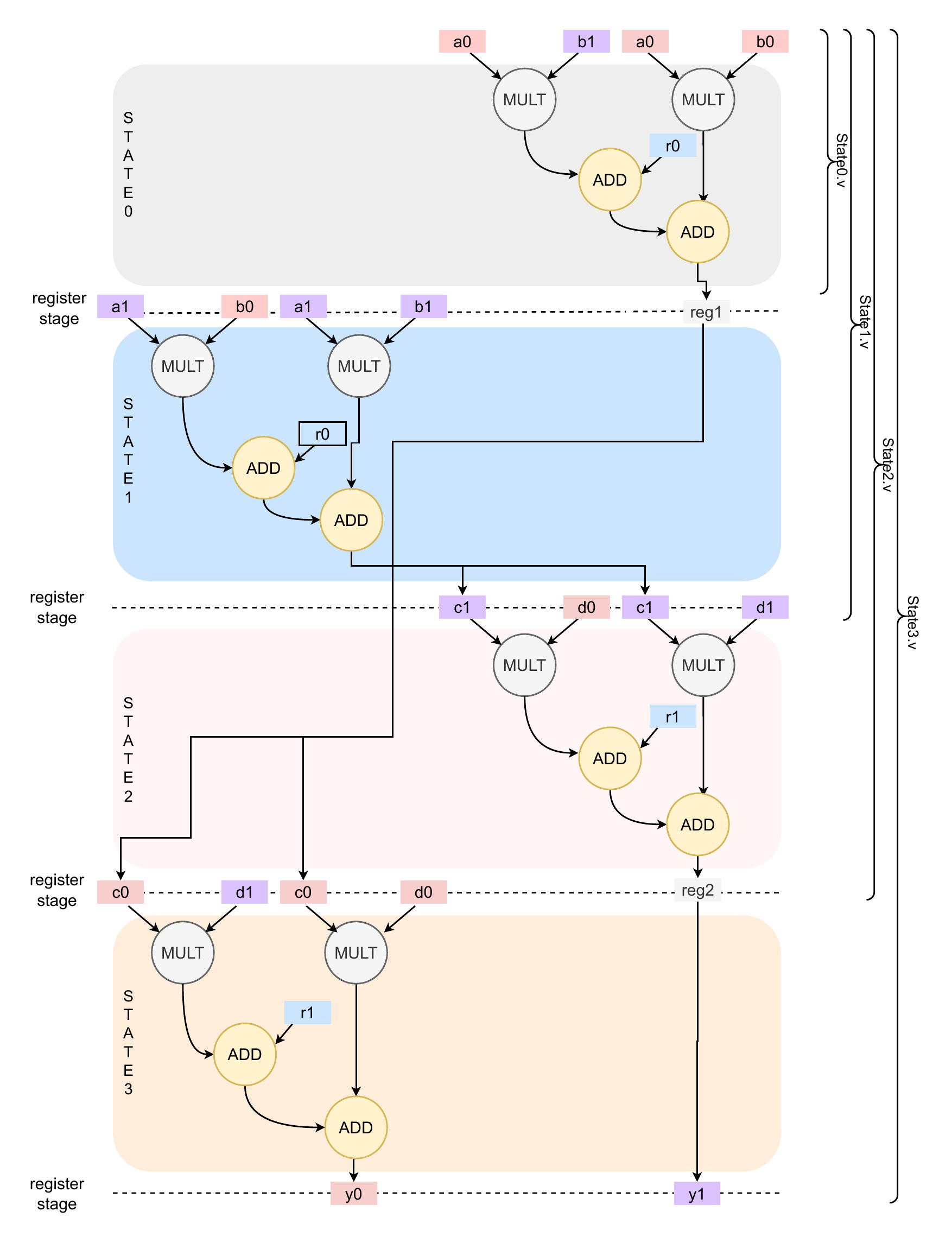}
    \caption{Intuition Example: Unrolled datapath over four-cycles of the cascaded masked multiplier of Figure \ref{fig:intuition0} corresponding to HLS-input of Listing \ref{lst:resourceallocation} (with a resource constraint of at-most two MULT (multiplier units) per cycle; depicted with operations corresponding to each state}
    \label{fig:intuition}
\end{figure}

\begin{table}[htbp] 
\centering
\caption{Registers and their corresponding labels}
\label{tab:intuitionlabelling}
\renewcommand{\arraystretch}{1.2}
\small{
\begin{tabular}{ll}
\toprule
\textbf{Variable Name} & \textbf{Label} \\
\midrule
$a0$ & $\{Share1\}$ \\
$a1$ & $\{Share1\}$ \\
$b0$ & $\{Share2\}$ \\
$b1$ & $\{Share2\}$ \\
$d0$ & $\{Share3\}$ \\
$d1$ & $\{Share3\}$ \\
$y0$ & $\{\}$ \\
$y1$ & $\{\}$ \\
$r0$ & $\{Mask\}$ \\
$r1$ & $\{Mask\}$ \\
\bottomrule
\end{tabular}}
\end{table}

To illustrate our state-wise splitting and the subsequent integration into our MaskedHLSVerif flow, consider the circuit performing two multiplication operations one after the other (out = $(a \times b)\times d)$) as shown in Figure \ref{fig:intuition}. This was obtained via unrolling the state-wise datapath of the output of Vitis HLS (consisting of four states)  for the HLS-input in Listing \ref{lst:resourceallocation}. Due to the resource-constraint of only two multipliers per cycle in the HLS input (line 11, 12 in Listing \ref{lst:resourceallocation}), the overall datapath consists of operations from the four FSM states, denoted STATE0, STATE1, STATE2 and STATE3, consisting of only two partial product (masked) multiplication operations per state.
The design has three input variables $a, b, d$ with two shares each $a0, a1, b0, b1, d0, d1$ each, indicating a first order masking. Two random values $r0, r1$ are used for re-sharing the partial products. In the datapath in Figure \ref{fig:intuition}, the first two states carries out the first masked multiplication and the final two states carries out the second masked multiplication. 

\textbf{Splitting into state-wise designs: }

For STATE0: the inputs are the primary inputs $a0, a1, b0, b1, r0$ and the output is $reg1$. The operations in the datapath would be : $temp0 = a0 \times b1; temp1 = a0 \times b0; temp2 = temp1 + r0; reg1 = temp2 + temp0;.$\footnote{The symbol $\times$ is used for a multiplicaton operation.} Thus, the datapath for $D_0$ consists of the operations in STATE0. Similarly, the datapath for $D_2$ consists of all the operations in STATE2 ($temp6 = c1 \times d0; temp7 =c1 \times d1; temp8 = temp6 + r1; reg2 = temp8 + temp7;$) and their dependent operations from the previous states ($temp3 = a1 \times b0; temp4 = a1 \times b1; temp5 = temp3 + r0; c1 = temp5 + temp4;$). For each state-wise design under verification, $\{D_0, D_1, D_2, D_3\}$, we include all dependent state operations from previous states, including the operations at that state and thus recursively obtain their datapath.

\textbf{Labelling the state-wise designs: }
The labelling of the input signals for the design are given in Table \ref{tab:intuitionlabelling}. The labels of the state-wise designs can be obtained directly form the labels of the design inputs that the state-wise design is dependent on. For example, for $D_2$, the inputs are: $d0, a1, b1, b0, r0, r1$, the labels of which are taken from Table \ref{tab:example}.

This labeling strategy results in the following advantages over a method that only includes the specific state operations for the design-under verification corresponding to a state $state_x$: 
\begin{itemize}
    \item The labeling strategy recursively defines the input of each state-wise design based on the original designs' inputs. Thus for each intermediate signal that is an input of the current state, the label of the corresponding overall design's input on which the intermediate signal is dependent is used as label.

         REBECCA requires the inputs of each design under verification to be labeled as a $Share$ of the $n^{th}$ input, a $Mask$ or a $Public$ variable. Splitting into state-based designs with only the specific state operations (and not all dependent computations from previous states) causes the intermediate operations of the original design to serve as the inputs of these state-wise designs. Obtaining the labels of these intermediate operations is thus necessary for verification. This, would require using the labels propagated through REBECCA to re-calculate the labels for each splitting of the original design based on the HLS output making verification complex.
    \item Compositional security: Since the valid operation for each state includes all the valid operations of the previous states that it is dependent on, therefore it ensures compositional verification by construction iff all the state-wise designs including the last state design passes verification. 

    In contrast, if the only the operations of each state would be included at a state-wise design under verification $state_x$, the compositional verification of the overall design is trivially undetermined.
    
\end{itemize}

\subsection{Verification of state-wise designs}
\label{subsec:statewiseverification}
Let the design in state $state_x$ be denoted as $D_x$. Once the $D_x$ are obtained, each of them are labeled with the same input labels as the original design $D$'s inputs (possible due to the recursive nature of the state-wise splitting proposed) creating the label files $L_x$. Thereafter, the label files $L_x$ and the state-wise designs $D_x$ are fed into REBECCA one at a time starting with $D_0$ and $L_0$ corresponding to the first state $state_0$.  For the verification given a security order:
\[
V(D_x) = 
\begin{cases}
1 & \text{if } D_x \text{ passes verification (secure) given } L_x \\
0 & \text{if } D_x \text{ fails verification (insecure) given } L_x
\end{cases}
\]
Given a state $state_x$ and state $state_{x+1}$ dependent on it, with the new state-wise verification flow presented, there could be the following possible outcomes of verification: 
\begin{itemize}
    \item Case1: $D_x$ design fails and $D_{x+1}$ design passes. In this case, the security flaw in design $D_{x+1}$ can be due to the following: 
        \begin{itemize}            
            \item An operation in $D_{x+1}$ that is independent of any operations on $D_{x}$.
        \end{itemize}
    \item Case2: $D_{x+1}$ design fails and $D_{x}$ design passes. In this case, the security flaw in design $D_{x}$ can be due to the following: 
    \begin{itemize}
        \item An operation in $D_{x}$ that does not have dependent operations in $D_{x+1}$.
       
    \end{itemize}
    \item Case3: Both $D_{x}$ and $D_{x+1}$ designs fail. In this case, the security flaw in the design $D_{x+1}$ can be due to the following:
    \begin{itemize}
        \item A dependent operation that failed in the design $D_{x}$.
        \item Two independent operations failing in designs $D_{x}$ and $D_{x+1}$.
        
    \end{itemize}
\end{itemize}

\subsection{Overall Flow}
\label{subsec:overallflow}
Figure \ref{fig:New_verification_flow} shows our proposed verification flow : MaskedHLSVerif.
\begin{figure}[htb]
    \centering
    \includegraphics[scale=0.5]{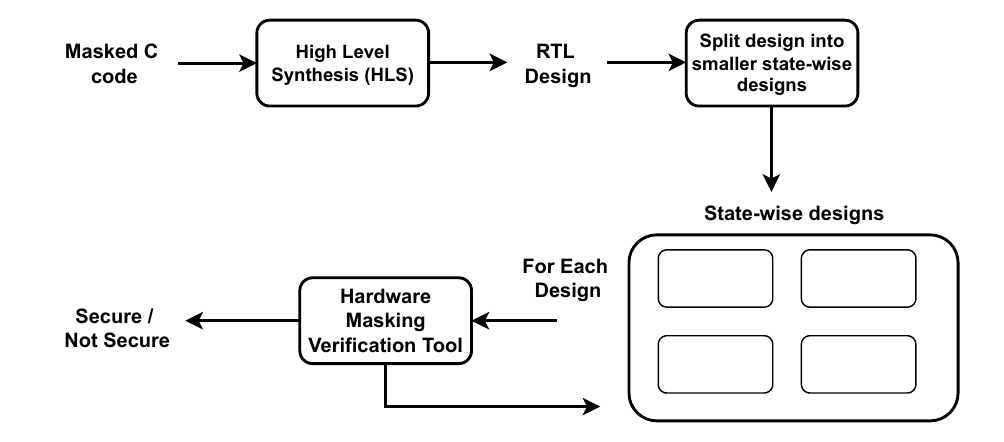}
    \caption{Proposed MaskedHLSVerif Verification flow (In this work, the Hardware Masking Verification Tool used is REBECCA \cite{rebeccabloem2018formal})}
    \label{fig:New_verification_flow}
\end{figure}
\footnotetext{In this work, the Hardware Masking Verification Tool used is REBECCA \cite{rebeccabloem2018formal}}
Below are the steps of our verification strategy :
\begin{enumerate}
    \item Generate equivalent RTL design from masked software design using HLS tool.
    \item Split the RTL design $D$ into smaller designs $D_x$ and their label files $L_x$ based on the operations being performed in each state as described previously.
    \item Starting at $D_0$ take each state-wise design $D_x$ and provide it as input to REBECCA for verification. As we are splitting the design, we also have to track every variable in the design and label them appropriately in the label file $L_x$.
    \item Repeat step 3 for every design created in step 2.
    \item If a design $D_x$ fails, we output the failure of the design according to the cases discussed above.
    \item If all the designs are secure, we conclude that the original RTL design is secure. Otherwise, it is not secure.
\end{enumerate}

% \subsection{Algorithm}
\label{subsec:algo}
A detailed algorithm for the proposed verification flow is presented Algorithm \ref{alg:verify_fsm}. 

% \textcolor{blue}{What do people write here?}
\begin{algorithm}[h]
\caption{MaskedHLSVerif: FSM-based masked RTL Verification}
\label{alg:verify_fsm}
\begin{algorithmic}[1]
\State \textbf{Input:} Masked RTL design $D$ from HLS, security order $d$
\State \textbf{Output:} Security status of $D$
\Procedure{VerifyRTL}{$D$}
    \State Generate RTL $D$ using HLS.
    \State Extract FSM states $\{S_1, S_2, \dots, S_n\}$ \Comment{e.g., ap\_CS\_fsm\_stateX}
    \For{each FSM state $S_i$}
        \State Extract operations and datapath in $S_i$ and all recursively dependent states
        \State Create sub-design $D_i$ with logic for $S_i$ and dependent states
        \State Track and label variables/signals in $D_i$
    \EndFor
    \For{each sub-design $D_i$}
        \State Run REBECCA on $D_i$ with security order $d$
        \If{REBECCA reports insecurity}
            \State \Return Insecure, $state_i$
        \EndIf
    \EndFor
    \State \Return Secure
\EndProcedure
\end{algorithmic}
\end{algorithm}
Given an input masked hardware design $D$ obtained using an HLS tool, the output of this algorithm gives whether the design is secure in the $d^{th}$ order.  If not, the locations of failure and the FSM state in which it fails are identified by our algorithm. 
% Initially, the FSM states are extracted from the HLS generated RTL. Following that, for each state the datapath, i.e., the operations corresponding to the state and the operations of the previous state that the current state is dependent on, the inputs and outputs of the state, the state-wise designs are built. Along with that the labels of all the inputs to the state, which are a subset of the primary inputs are replicated and stored. This gives two files $D_x$ corresponding to the design, and $L_x$ corresponding to the labels of $state_x$. With the files for all states available, REBECCA is invoked on them one at a time starting at the initial state. At the end of the algorithm, if all the states returned secure, then the overall design is secure. Otherwise the insecure designs and their traces from REBECCA corresponding to the locations of failure are reported. 

\subsection{Soundness of MaskedHLSVerif}

To establish the soundness of our verification flow, we start by defining the following: 
\begin{itemize}

\item \textbf{Ordering of states:} For two states $state_x$ and $state_y$ of a design $D$, if $x < y$, then operations in $state_x$ have dependent operations in $state_y$ or not, but operations in $state_y$ do not have dependent operations in $state_x$.

\item \textbf{Trace: }
An execution \textit{trace} of the design under verification $D$ is a sequence of
datapath operations executed from primary inputs to primary outputs under a single execution of the controller FSM. 

% \item \textbf{State-local trace: }
% A \textit{state-local trace} $T_x$ is the sequence of datapath operations
% executed while the controller FSM resides in a particular state $state_x$.

% For a given execution of the controller FSM, a trace $T$ can be written as
% \[
% T = T_{x_1} \circ T_{x_2} \circ \cdots \circ T_{x_k},
% \]
% where $(state_{x_1}, state_{x_2}, \ldots, state_{x_k})$ is the sequence of FSM states visited during that execution, with $state_{x_1}$ being the initial state
% and $state_{x_k}$ the final state of the design.

\end{itemize}

 We assume that the REBECCA algorithm is sound for a fixed-input datapath. To prove that MaskedHLSVerif verifies HLS-generated RTL correctly we proceed via proving in two steps:
\begin{itemize}
    \item Firstly, we prove that the state-wise split is sound.
    \item Then we prove that the state wise verification is sound.
\end{itemize}
\newtheorem{lemma}{Lemma}
\newtheorem{theorem}{Theorem}
\newtheorem{proposition}{Proposition}
\newtheorem{corollary}{Corollary}

\begin{lemma}[Soundness of State-wise Split]\label{lemma:split}

The state-wise splitting constructs state-wise designs $D_0, D_1, \ldots, D_n$ such that for any execution trace $T$ of the design $D$, every operation occurring in $T$ is captured by at least one of the state-wise designs $D_x$, $x \in \{0, n\}$.
\end{lemma}

\begin{proof}
% \textit{Proof Idea}\\
\textbf{Base Case:} 
For the initial FSM state $state_0$, by construction, the corresponding state-wise design $D_0$ contains all datapath operations enabled while the FSM is in $state_0$. 

\noindent \textbf{Inductive Case:} 
% Due to the unrolling
Assume that for any FSM state $state_j$ with $j < n$, the corresponding state-wise design $D_j$ captures all operations executed while the FSM is in $state_j$ and operations of all previous states $state_j$ is dependent on. 
Let a state $state_{j+1}$ immediately follow $state_{j}$,
%($T = T_0 \circ T_1 \circ \ldots T_{j-1} \circ T_j \circ T_{j+1} \circ \ldots \circ T_n$), 
\begin{itemize}

\item The state-wise design $D_{j+1}$ is constructed by including all datapath operations enabled in $state_{j+1}$, along with all recursively dependent operations from $state_{j}$. 
\item In addition, the inductive hypothesis states that $D_j$ captures all operations executed while the FSM is in $state_j$ and operations of all previous states $state_j$ are dependent on. Thus, $state_{j+1}$ satisfies the inductive hypothesis.
    
\end{itemize}
\noindent \textbf{Conclusion:} Thus, from the above points, follows that the design $D_n$ corresponding to the last state captures all operations occurring in a trace $T$ of one single execution of the controller FSM.
\end{proof}

\begin{lemma}[Soundness of MaskedHLSVerif]\label{lemma:soundness}   

\begin{itemize}
    \item \textbf{Assumption 1: }REBECCA is sound for a fixed-input datapath.
    \item \textbf{Assumption 2: }The splitting of the design into state-wise designs is sound.
\end{itemize}
\noindent Given Assumption 1 and Assumption 2, Algorithm 1 will result in an output $true$ if there is at least one intermediate wire that leaks and $false$ otherwise via the power side channel in the $1^{st}$ security order.
\end{lemma}

\begin{proof}
%The goal is to demonstrate that Algorithm 1 accurately produces the results of the power-side-channel security evaluation of the design, using the correct set of control signals for the MUX logic for the states. This means that, if 
We need to prove that all the intermediate computation result captured by each sub-design $D_x$ belongs to $D$ and each (intended) intermediate computation result in $D$ belongs to at least one of $D_x$, we are done.
From the previous proof, it follows that for each trace $T$ of $D$, any operation in $T$ is captured by at least one state-wise design $D_x$; therefore, each valid intermediate computation of $D$ is captured by at least one of $D_x$ and vice versa. Hence, the soundness of Algorithm \ref{alg:verify_fsm} follows.
\end{proof}

% \section{Results}
%%%%%%%%%%%%%%%%%%%%%%%%%%%%%%%%%%%%%%%%%%%%%%%%%%%%%%%%%%%%%%%%%%%%%%%%%%%%%%%%
\section{Results}
\label{sec:results}
In this section we present the results of using our proposed flow to verify the benchmark designs generated by Vitis HLS tool version 2022.2. We have used six benchmarks: DOM cascaded version4, COMAR cascaded, HPC1 cascaded, HPC2 cascaded, PRESENT DOM and, PRESENT HPC1, which were briefly introduced in Table \ref{tab:motivation}, which we describe in the next Subsection.
\subsection{Benchmark Details}
The benchmarks used to evaluate our verification flow consists of four cascaded masked multipliers, performing two multiplication operations one after another, using four different masking schemes: DOM \cite{dom2016}, HPC1 \cite{HPC2020}, HPC2 \cite{HPC2020} and COMAR \cite{comar}. In addition we used the PRESENT \cite{present2007} cipher's S-box with AND depth 2 and 4 AND gates from the work of Cassiers et. al. \cite{HPC2020} masked using DOM and HPC1. All these designs were converted from their C-specifications to RTL via Vitis HLS subject to a set of resource/timing constraints. We list the details of HLS synthesis of these designs using Vitis HLS below:
\begin{itemize}
    \item DOM cascaded masked multiplier: synthesized by Vitis HLS with target clock set as 1ns resulting in a 4-state design.
    \item HPC1 cascaded masked multiplier: synthesized by Vitis HLS with target clock set as 0.5ns resulting in a 4 state design.
    \item HPC2 cascaded masked multiplier: synthesized by Vitis HLS with target clock set as 0.5ns resulting in a 3-state design.
    \item COMAR cascaded masked multiplier: synthesized by Vitis HLS with target clock set as 0.5ns resulting in a 3-state design.
    \item PRESENT masked design using DOM: synthesized by Vitis HLS with target clock set as 2ns resulting in a 2-state design.
    \item PRESENT masked design using HPC1: synthesized by Vitis HLS with target clock set as 2ns resulting in a 2-state design.
\end{itemize}

They are further split into smaller designs based on the operations being performed in each state. For DOM and COMAR cascaded designs, FSM has four states. So, four state-wise have been created from the original DOM and COMAR RTL code generated from HLS using a masked software implementation of DOM. Similarly, in HPC1 and HPC2 designs, FSM has three states resulting in the generation of three smaller designs. For PRESENT DOM and PRESENT HPC designs, FSM has two states, which results in two smaller designs. On these benchmarks we perform two experiments to validate our verification flow. In Experiment 1, we aim to verify the HLS generated RTLs with our verification method. It may be recalled that REBECCA produces false-positive results for each of theses above six test scenarios. In Experiment 2, we demonstrate that our tool is able to catch a masking flaw produced by Vitis HLS in synthesizing the cascaded DOM masked multiplier. 

\begin{table*}[htbp]
\centering
\caption{Metrics for different states of DOM RTL (S: stable, T: transient)}
\label{table:metrics_DOM}
\renewcommand{\arraystretch}{1.2}

\begin{adjustbox}{max width=\textwidth}
\begin{tabular}{c cc c c c cc cc cc cc cc}
\toprule
\multirow{2}{*}{\textbf{State}} 
& \multicolumn{2}{c}{\textbf{Exec. Time (s)}} 
& \multirow{2}{*}{\textbf{LOC (C)}} 
& \multirow{2}{*}{\textbf{LOC (RTL)}} 
& \multirow{2}{*}{\textbf{LOC (JSON)}} 
& \multicolumn{2}{c}{\textbf{Vars}} 
& \multicolumn{2}{c}{\textbf{Assertions}} 
& \multicolumn{2}{c}{\textbf{Digraph}} 
& \multicolumn{2}{c}{\textbf{Expected}} 
& \multicolumn{2}{c}{\textbf{Actual}} \\
\cmidrule(lr){2-3}\cmidrule(lr){7-8}\cmidrule(lr){9-10}
\cmidrule(lr){11-12}\cmidrule(lr){13-14}\cmidrule(lr){15-16}
& S & T & & & & S & T & S & T & Nodes & Edges & S & T & S & T \\
\midrule
State 1 & 0.09 & 0.07 & 33 & 50  & 498  & 90  & 165 & 18 & 33 & 15 & 16 & True & True & True & True \\
State 2 & 0.06 & 0.08 & 33 & 135 & 792  & 114 & 209 & 22 & 41 & 19 & 22 & True & True & True & True \\
State 3 & 0.09 & 0.12 & 33 & 164 & 1152 & 174 & 319 & 32 & 61 & 29 & 38 & True & True & True & True \\
State 4 & 0.10 & 0.13 & 33 & 189 & 1362 & 186 & 341 & 34 & 65 & 31 & 42 & True & True & True & True \\
\bottomrule
\end{tabular}
\end{adjustbox}
\end{table*}

\begin{table*}[htbp]
\centering
\caption{Metrics for different states of COMAR RTL (S: stable, T: transient)}
\label{table:metrics_COMAR}
\renewcommand{\arraystretch}{1.2}

\begin{adjustbox}{max width=\textwidth}
\begin{tabular}{c cc c c c cc cc cc cc cc}
\toprule
\multirow{2}{*}{\textbf{State}} 
& \multicolumn{2}{c}{\textbf{Exec. Time (s)}} 
& \textbf{LOC (C)} & \textbf{LOC (RTL)} & \textbf{LOC (JSON)}
& \multicolumn{2}{c}{\textbf{Vars}}
& \multicolumn{2}{c}{\textbf{Assertions}}
& \multicolumn{2}{c}{\textbf{Digraph}}
& \multicolumn{2}{c}{\textbf{Expected}}
& \multicolumn{2}{c}{\textbf{Actual}} \\
\cmidrule(lr){2-3}\cmidrule(lr){7-8}\cmidrule(lr){9-10}
\cmidrule(lr){11-12}\cmidrule(lr){13-14}\cmidrule(lr){15-16}
& S & T & & & & S & T & S & T & Nodes & Edges & S & T & S & T \\
\midrule
State 1 & 0.08 & 0.10 & 90 & 81  & 690  & 168 & 315 & 24 & 45 & 21 & 20 & True & True & True & True \\
State 2 & 0.31 & 0.41 & 90 & 168 & 1905 & 649 & 1239 & 62 & 121 & 59 & 69 & True & True & True & True \\
State 3 & 0.43 & 0.53 & 90 & 207 & 2421 & 836 & 1596 & 79 & 155 & 76 & 96 & True & True & True & True \\
State 4 & 0.45 & 0.60 & 90 & 238 & 2715 & 935 & 1785 & 88 & 173 & 85 & 108 & True & True & True & True \\
\bottomrule
\end{tabular}
\end{adjustbox}
\end{table*}

\begin{table*}[htbp]
\centering
\caption{Metrics for different states of HPC1 RTL (S: stable, T: transient)}
\label{table:metrics_hpc1}
\renewcommand{\arraystretch}{1.2}

\begin{adjustbox}{max width=\textwidth}
\begin{tabular}{c cc c c c cc cc cc cc cc}
\toprule
\multirow{2}{*}{\textbf{State}} 
& \multicolumn{2}{c}{\textbf{Exec. Time (s)}} 
& \textbf{LOC (C)} & \textbf{LOC (RTL)} & \textbf{LOC (JSON)}
& \multicolumn{2}{c}{\textbf{Vars}}
& \multicolumn{2}{c}{\textbf{Assertions}}
& \multicolumn{2}{c}{\textbf{Digraph}}
& \multicolumn{2}{c}{\textbf{Expected}}
& \multicolumn{2}{c}{\textbf{Actual}} \\
\cmidrule(lr){2-3}\cmidrule(lr){7-8}\cmidrule(lr){9-10}
\cmidrule(lr){11-12}\cmidrule(lr){13-14}\cmidrule(lr){15-16}
& S & T & & & & S & T & S & T & Nodes & Edges & S & T & S & T \\
\midrule
State 1 & 0.03 & 0.03 & 48 & 47  & 240  & 28  & 49  & 10 & 17 & 7  & 6  & True & True & True & True \\
State 2 & 0.10 & 0.13 & 48 & 134 & 966  & 200 & 375 & 28 & 53 & 25 & 30 & True & True & True & True \\
State 3 & 0.16 & 0.22 & 48 & 204 & 1495 & 324 & 612 & 39 & 75 & 36 & 48 & True & True & True & True \\
\bottomrule
\end{tabular}
\end{adjustbox}
\end{table*}

\begin{table*}[htbp]
\centering
\caption{Metrics for different states of HPC2 RTL (S: stable, T: transient)}
\label{table:metrics_hpc2}
\renewcommand{\arraystretch}{1.2}

\begin{adjustbox}{max width=\textwidth}
\begin{tabular}{c cc c c c cc cc cc cc cc}
\toprule
\multirow{2}{*}{\textbf{State}} 
& \multicolumn{2}{c}{\textbf{Exec. Time (s)}} 
& \textbf{LOC (C)} & \textbf{LOC (RTL)} & \textbf{LOC (JSON)}
& \multicolumn{2}{c}{\textbf{Vars}}
& \multicolumn{2}{c}{\textbf{Assertions}}
& \multicolumn{2}{c}{\textbf{Digraph}}
& \multicolumn{2}{c}{\textbf{Expected}}
& \multicolumn{2}{c}{\textbf{Actual}} \\
\cmidrule(lr){2-3}\cmidrule(lr){7-8}\cmidrule(lr){9-10}
\cmidrule(lr){11-12}\cmidrule(lr){13-14}\cmidrule(lr){15-16}
& S & T & & & & S & T & S & T & Nodes & Edges & S & T & S & T \\
\midrule
State 1 & 0.05 & 0.05 & 82 & 90  & 577  & 78  & 143 & 16 & 29 & 13 & 12 & True & True & True & True \\
State 2 & 0.11 & 0.14 & 82 & 205 & 1498 & 210 & 390 & 33 & 63 & 30 & 38 & True & True & True & True \\
State 3 & 0.15 & 0.20 & 82 & 332 & 2364 & 280 & 520 & 43 & 83 & 40 & 58 & True & True & True & True \\
\bottomrule
\end{tabular}
\end{adjustbox}
\end{table*}

\begin{table*}[htbp]
\centering
\caption{Metrics for different states of PRESENT DOM RTL (S: stable, T: transient)}
\label{table:metrics_present_dom}
\renewcommand{\arraystretch}{1.2}

\begin{adjustbox}{max width=\textwidth}
\begin{tabular}{c cc c c c cc cc cc cc cc}
\toprule
\multirow{2}{*}{\textbf{State}} 
& \multicolumn{2}{c}{\textbf{Exec. Time (s)}} 
& \textbf{LOC (C)} & \textbf{LOC (RTL)} & \textbf{LOC (JSON)}
& \multicolumn{2}{c}{\textbf{Vars}}
& \multicolumn{2}{c}{\textbf{Assertions}}
& \multicolumn{2}{c}{\textbf{Digraph}}
& \multicolumn{2}{c}{\textbf{Expected}}
& \multicolumn{2}{c}{\textbf{Actual}} \\
\cmidrule(lr){2-3}\cmidrule(lr){7-8}\cmidrule(lr){9-10}
\cmidrule(lr){11-12}\cmidrule(lr){13-14}\cmidrule(lr){15-16}
& S & T & & & & S & T & S & T & Nodes & Edges & S & T & S & T \\
\midrule
State 1 & 0.42 & 0.88 & 82 & 342 & 3438 & 870 & 1653 & 90 & 177 & 87 & 148 & True & True & True & True \\
State 2 & 0.45 & 0.97 & 82 & 408 & 3942 & 950 & 1805 & 98 & 193 & 95 & 156 & True & True & True & True \\
\bottomrule
\end{tabular}
\end{adjustbox}
\end{table*}

\begin{table*}[htbp]
\centering
\caption{Metrics for different states of PRESENT HPC1 RTL (S: stable, T: transient)}
\label{table:metrics_present_hpc1}
\renewcommand{\arraystretch}{1.2}

\begin{adjustbox}{max width=\textwidth}
\begin{tabular}{c cc c c c cc cc cc cc cc}
\toprule
\multirow{2}{*}{\textbf{State}} 
& \multicolumn{2}{c}{\textbf{Exec. Time (s)}} 
& \textbf{LOC (C)} & \textbf{LOC (RTL)} & \textbf{LOC (JSON)}
& \multicolumn{2}{c}{\textbf{Vars}}
& \multicolumn{2}{c}{\textbf{Assertions}}
& \multicolumn{2}{c}{\textbf{Digraph}}
& \multicolumn{2}{c}{\textbf{Expected}}
& \multicolumn{2}{c}{\textbf{Actual}} \\
\cmidrule(lr){2-3}\cmidrule(lr){7-8}\cmidrule(lr){9-10}
\cmidrule(lr){11-12}\cmidrule(lr){13-14}\cmidrule(lr){15-16}
& S & T & & & & S & T & S & T & Nodes & Edges & S & T & S & T \\
\midrule
State 1 & 0.52 & 1.17 & 80 & 515 & 3788 & 1056 & 2016 & 99 & 195 & 96 & 164 & True & True & True & True \\
State 2 & 0.57 & 1.31 & 80 & 578 & 4293 & 1144 & 2184 & 107 & 211 & 104 & 172 & True & True & True & True \\
\bottomrule
\end{tabular}
\end{adjustbox}
\end{table*}

\subsection{Experiment 1}
%: Using proposed verification flow of REBECCA on benchmarks generated via Vivado HLS}
The first experiment uses the proposed verification flow of REBECCA on benchmarks generated via Vitis HLS. For the cascaded DOM multiplier design, the results corresponding to each state have been shown in the Table \ref{table:metrics_DOM} across various metrics like the number of variables being used, number of assertions being created, lines of code in c, RTL as well as in JSON which represents gate-level design for the hardware implementation and is created in intermediate steps of verification by REBECCA.
We have also listed the number of nodes and edges in the digraph that is being generated to verify the design in REBECCA. this graph is created by removing redundant circuit components and only tracks the paths where sensitive data will be propagated. Similar results across various metrics for all states of cascaded COMAR, cascaded HPC1, cascaded HPC2, PRESENT DOM-masked, and PRESENT HPC1-masked have been shown in the Tables \ref{table:metrics_COMAR}, \ref{table:metrics_hpc1}, \ref{table:metrics_hpc2}, \ref{table:metrics_present_dom} and \ref{table:metrics_present_hpc1}, respectively. 

As we can observe from the results shown in the tables \ref{table:metrics_DOM}, \ref{table:metrics_COMAR}, \ref{table:metrics_hpc1}, \ref{table:metrics_hpc2}, \ref{table:metrics_present_dom} and \ref{table:metrics_present_hpc1} for the designs DOM, COMAR, HPC1, HPC2, PRESENT with DOM and PRESENT with HPC1 respectively, with each state, the number of variables and assertions increase.
This, in turn, increases the complexity of verifying the state-wise designs.
%A consolidated set of results for all the six benchmarks is presented in Table \ref{table:results_new} which shows the aggregate execution times, clock period, number of states, number of variables and number of assertions for all the designs. 
We can also observe the inverse relationship between the clock period and the number of states generated in FSM.
Consider the HLS-generated RTL and the result of verifying them using REBECCA and our tool flow, MaskedHLSVerif in Table \ref{tab:resultsexpt2}. 
The results shows that MaskedHLSVerif can formally verify the HLS-generated RTL using Vitis HLS which registered false positives on REBECCA.
\begin{table}[htbp]
\centering
\caption{Comparison of PSCA Security Verification Results: Rebecca vs. New Verification Flow}
\label{tab:resultsexpt2}
\renewcommand{\arraystretch}{1.2}

\begin{adjustbox}{max width=\columnwidth}
\begin{tabular}{l cc cc cc cc}
\toprule
\multirow{2}{*}{\textbf{Example}} 
& \multirow{2}{*}{\textbf{FSM}} 
& \multirow{2}{*}{\textbf{Reg}} 
& \multicolumn{2}{c}{\textbf{Expected}} 
& \multicolumn{2}{c}{\textbf{Rebecca}} 
& \multicolumn{2}{c}{\textbf{MaskedHLSVerif}} \\
\cmidrule(lr){4-5}
\cmidrule(lr){6-7}
\cmidrule(lr){8-9}
& & & $s$ & $t$ & $s$ & $t$ & $s$ & $t$ \\
\midrule
DOM RTL version4     & \cmark & \cmark & True  & True  
& \textcolor{red}{False} & \textcolor{red}{False} 
& \textcolor{green}{True} & \textcolor{green}{True} \\
COMAR RTL            & \cmark & \cmark & True  & True  
& \textcolor{red}{False} & \textcolor{red}{False} 
& \textcolor{green}{True} & \textcolor{green}{True} \\
HPC1 RTL             & \cmark & \cmark & True  & True  
& \textcolor{red}{False} & \textcolor{red}{False} 
& \textcolor{green}{True} & \textcolor{green}{True} \\
HPC2 RTL             & \cmark & \cmark & True  & True  
& \textcolor{red}{False} & \textcolor{red}{False} 
& \textcolor{green}{True} & \textcolor{green}{True} \\
PRESENT DOM RTL      & \cmark & \cmark & True  & True  
& \textcolor{red}{False} & \textcolor{red}{False} 
& \textcolor{green}{True} & \textcolor{green}{True} \\
PRESENT HPC1 RTL     & \cmark & \cmark & True  & True  
& \textcolor{red}{False} & \textcolor{red}{False} 
& \textcolor{green}{True} & \textcolor{green}{True} \\
\bottomrule
\end{tabular}
\end{adjustbox}
\end{table}

% \begin{table}[htbp]
% \centering
% \caption{Results obtained for PSCA security vulnerability}
% \label{tab:resultsexpt2}
% \renewcommand{\arraystretch}{1.2}
% \begin{adjustbox}{max width=\columnwidth}
% \begin{tabular}{l cc cc cc}
% \toprule
% \multirow{2}{*}{\textbf{Example}} 
% & \multirow{2}{*}{\textbf{FSM}} 
% & \multirow{2}{*}{\textbf{Reg}} 
% & \multicolumn{2}{c}{\textbf{Expected Op}} 
% & \multicolumn{2}{c}{\textbf{Rebecca + New Verification Flow Op}} \\
% \cmidrule(lr){4-5}
% \cmidrule(lr){6-7}
% & & & $s$ & $t$ & $s$ & $t$ \\
% \midrule
% DOM RTL version4 & \cmark & \cmark & True & True & \textcolor{green}{True} & \textcolor{green}{True} \\
% COMAR RTL        & \cmark & \cmark & True & True & \textcolor{green}{True} & \textcolor{green}{True} \\
% HPC1 RTL         & \cmark & \cmark & True & True & \textcolor{green}{True} & \textcolor{green}{True} \\
% HPC2 RTL         & \cmark & \cmark & True & True & \textcolor{green}{True} & \textcolor{green}{True} \\
% PRESENT DOM RTL  & \cmark & \cmark & True & True & \textcolor{green}{True} & \textcolor{green}{True} \\
% PRESENT HPC1 RTL & \cmark & \cmark & True & True & \textcolor{green}{True} & \textcolor{green}{True} \\
% \bottomrule
% \end{tabular}
% \end{adjustbox}
% \end{table}
\noindent\textbf{Security Validation via TVLA.}
To verify that the HLS-generated RTL is indeed PSCA-secure, we evaluated all designs using a standard TVLA flow. RTL generated by Vitis HLS was synthesized to gate-level using Synopsys Design Compiler, simulated using VCD dumps, and analyzed with Synopsys PrimeTime. Across up to ten million traces with fixed and random plaintexts, all designs exhibited t-test values within the $\pm 4.5$ threshold, confirming first-order security. This validates that the RTL itself is secure and that the failures reported by \textsc{REBECCA} are false positives.
\subsection{Experiment 2}
%Forcefully inducing a PSCA security vulnerability and testing if the proposed verification flow still works}
\begin{figure*}
    \centering
    \includegraphics[width=1\linewidth]{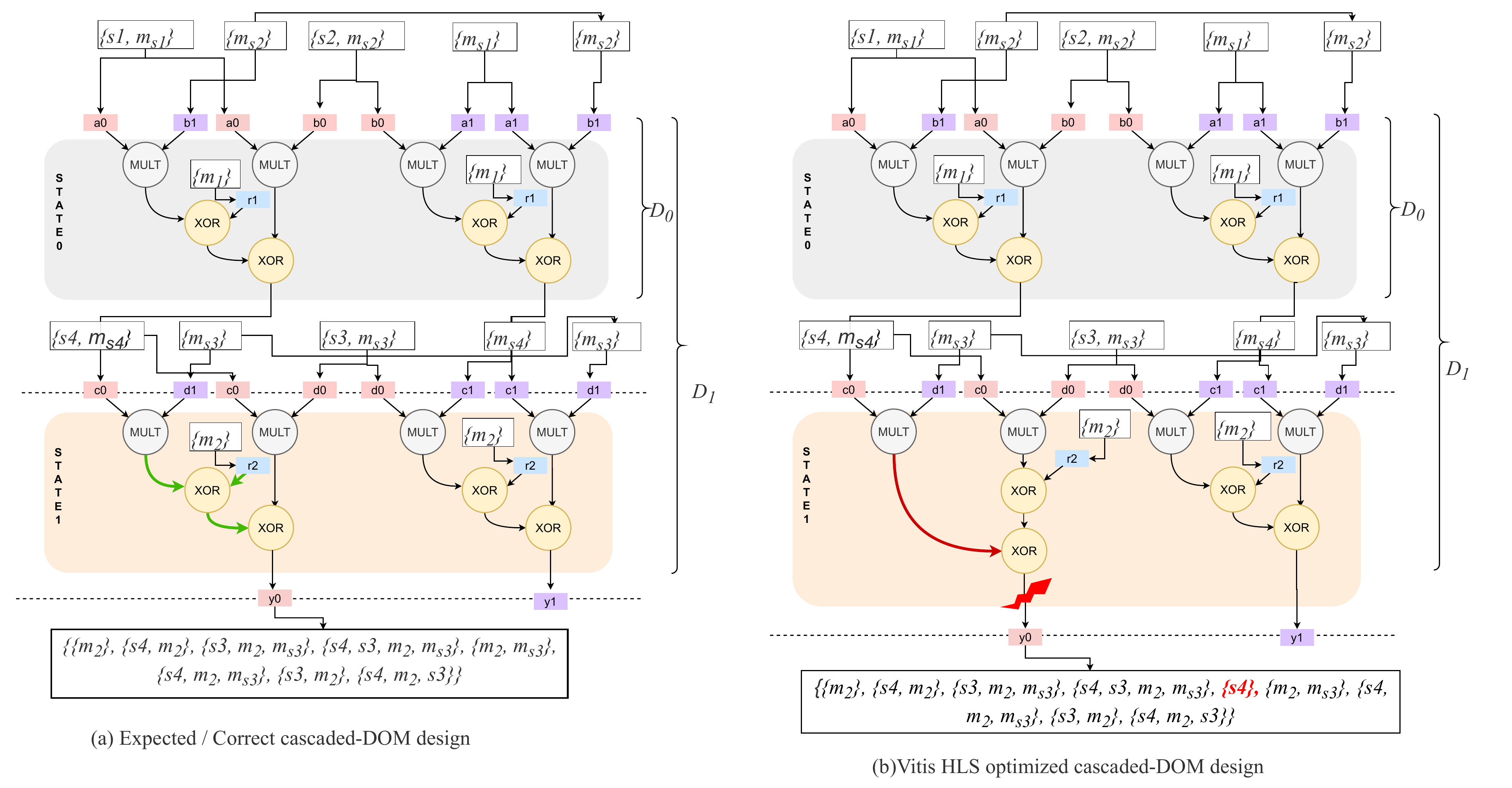}
    \caption{Unrolled Datapath of (a) Correct (unoptimized) cascaded-DOM RTL vs (b) Vitis HLS optimized cascaded-DOM RTL labelled using static (S) rules in using our REBECCA flow demonstrating that the security-impact of optimization are identified by our flow}
    \label{fig:reassociatedlabels}
\end{figure*}
% \textcolor{red}{Write about setup}
Previously in Section \ref{sec:impactofhls}, we had discussed how HLS optimizations may hamper PSCA-security guaranteed by masking by altering or removing the masking related operations in the design. One such optimization is reassociation via expression-balancing that is used by HLS to reduce the depth of the computation tree. 
This optimization alters the sub-epxression evaluation sequence for the output expressions. For example, the cascaded-DOM multiplication design which calculates the product $a \& b \& c$ as a two share masked version $c0, c1 = dom(a0, a1, b0, b1, r1);$ $y0, y1 = dom(c0, c1, d0, d1, r2)$ consists of two multiplication operations masked by two DOM gates one after the other. The second multiplication for example should proceed via the following parenthesized order: $((c0 \otimes d1) \oplus r2) \oplus (c0 \otimes d0)$. A masking flaw can be induced when a reassociation optimization reassociates the mask $r2$ in the masking of the product term $(c0 \otimes d1)$ to the product term $(c0 \otimes d0)$ following the order: $((c0 \otimes d0) \oplus r2) \oplus (c0 \otimes d1)$ resulting in a first order masking flaw in the design. 

\textit{Setup: }In order to check whether our masking verification method can catch such a masking flaw due to re-association, we cause the pragma in Vitis 2022.2 called $pragma\_expression\_balancing$ to be ON within the C code of the cascaded DOM masked multiplication which is the input. The target clock set as 1ns resulting in a $2$ state design. The first state is the masked multiplication of $a$ and $b$. The second state multiplies the result of the first state with $c$. Both the states uses the dom multiplication gadgets with random values $r1$ and $r2$, respectively. We extract the designs corresponding to $state_0$ and $state_1$ according to our splitting procedure in Section \ref{subsec:splitting} and apply our MaskedHLSVerif verification flow, outlined in Algorithm \ref{alg:verify_fsm}, on the statewise designs in the static state. From the verification results in Table \ref{table:metrics_cascadeddom_expressionbal}, we can see that the first state passes the verification but the second state fails the verification. This was caused due to the flaw due to re-association caused $pragma\_expression\_balancing$ ON, which our verification flow correctly catches as a  masking flaw. Figure \ref{fig:reassociatedlabels} demonstrates the difference in labels of the outputs $y0$ in the correct design versus the incorrect reassociated design obtained via Vitis in the above experiment. It can be seen that the labels of $y0$ in the optimized design (Figure \ref{fig:reassociatedlabels}b) has a unmasked secret values in its label set : $\{s4\}$, in addition to the labels in Figure \ref{fig:reassociatedlabels}a, justifying the masking flaw due to re-association at the output $y0$. This demonstrates that, in addition to correctly verifying FSM based RTL output of HLS, our verification flow is able to catch HLS-optimization induced masking flaws in HLS-generated masked hardware.
\begin{table*}[htbp]
\centering
\caption{Metrics for different states of cascaded-DOM RTL obtained via expression balancing optimization in Vitis HLS (S: stable, T: transient)}
\label{table:metrics_cascadeddom_expressionbal}
\renewcommand{\arraystretch}{1.2}

\begin{adjustbox}{max width=\textwidth}
\begin{tabular}{c cc c c c cc cc cc cc cc}
\toprule
\multirow{2}{*}{\textbf{State}} 
& \multicolumn{2}{c}{\textbf{Exec. Time (s)}} 
& \textbf{LOC (C)} 
& \textbf{LOC (RTL)} 
& \textbf{LOC (JSON)} 
& \multicolumn{2}{c}{\textbf{Vars}} 
& \multicolumn{2}{c}{\textbf{Assertions}} 
& \multicolumn{2}{c}{\textbf{Digraph}} 
& \multicolumn{2}{c}{\textbf{Expected}} 
& \multicolumn{2}{c}{\textbf{Actual}} \\

\cmidrule(lr){2-3}
\cmidrule(lr){7-8}
\cmidrule(lr){9-10}
\cmidrule(lr){11-12}
\cmidrule(lr){13-14}
\cmidrule(lr){15-16}

& S & T &  &  &  & S & T & S & T & Nodes & Edges & S & T & S & T \\
\midrule
State 1 & 0.08 & 0.07 & 28 & 80  & 588  & 102 & 187 & 20 & 37 & 17 & 20 & True & True & True & True \\
State 2 & 0.28 & 0.73 & 28 & 121 & 1454 & 405 & 765 & 48 & 93 & 45 & 67 & \textcolor{green}{False} & \textcolor{green}{False} & \textcolor{green}{False} & \textcolor{green}{False} \\
\bottomrule
\end{tabular}
\end{adjustbox}
\end{table*}

\section{Conclusion}
\label{sec:conclusion}

This paper introduces MaskedHLSVerif: a method to verify the masking security of hardware obtained from the masked software implementation of cryptographic algorithms using High-level Synthesis. We first show the limitation of the state-of-the-art masking verification technique in verifying resource shared datapath and a controller FSM. HLS, in particular, generates such design based on the target clock period and latency. We then introduce the concept of state wise active datapath verification to verify such hardware using existing verification method.  This is the first work that targets verification of high-level synthesized masked hardware. Experiments with six cryptographic designs confirm the usefulness of the proposed method. The impact of HLS optimizations on PSCA security is largely unexplored. It would be interesting to explore such an impact with MaskedHLSVerif in the future. 
%In general, masking verification techniques are not scalable. In future, we aim to work on the scalability of hardware masking verification. 

%%% Local Variables:
%%% mode: latex
%%% TeX-master: "main"
%%% End:

\bibliographystyle{ieeetr}
\bibliography{ref}
\end{document}